\documentclass[nofootinbib,onecolumn ,aps,longbibliography,a4paper,superscriptaddress,tightenlines,notitlepage,11pt]{revtex4-1}

\usepackage{graphicx}
\usepackage{amsmath}
\usepackage{amssymb}
\usepackage{amsfonts}
\usepackage{amsthm}
\usepackage{mathtools}
\usepackage{hyperref}
\usepackage{xcolor}
\usepackage[T1]{fontenc}
\usepackage[english]{babel}
\usepackage{braket}
\usepackage{enumitem}
\usepackage{tikz}

\usepackage{times}

\hypersetup{colorlinks=true,urlcolor=[rgb]{0,0,0.5},citecolor=[rgb]{0.5,0,0},linkcolor=[rgb]{0,0,0.4}}

\usepackage[bbgreekl]{mathbbol}
\DeclareSymbolFontAlphabet{\mathbb}{AMSb}
\DeclareSymbolFontAlphabet{\mathbbl}{bbold}

\newtheorem{theorem}{Theorem}
\newtheorem{corollary}{Corollary}
\newtheorem{lemma}{Lemma}

\newtheorem{definition}{Definition}
\newtheorem{conjecture}{Conjecture}
\newtheorem{remark}{Remark}

\def\autorefapp#1{\hyperref[#1]{Appendix~\ref{#1}}}

\DeclareMathOperator\tr{Tr}
\DeclareMathOperator{\Tr}{Tr}

\DeclareMathOperator{\spann}{span}
\def\ii{\mathrm{i}}
\def\d{\mathrm{d}}
\newcommand{\C}{\mathbb{C}}

\def\ep{\varepsilon}

\def\iden{\mathbbl{1}}
\def\ketbra#1{ |{#1}\rangle\!\langle{#1}| }
\def\vev#1{\langle{#1}\rangle}
\def\ni{\noindent}
\def\nn{\nonumber\\}

\def\Wg{{\cal W}\! g}
\def\orth{{\rm O}}

\begin{document}

\title{Improved spectral gaps for random quantum circuits:\\ large local dimensions and all-to-all interactions}

\author{Jonas Haferkamp}
\affiliation{Dahlem Center for Complex Quantum Systems, Freie Universit{\"a}t Berlin, Germany}
\author{Nicholas Hunter-Jones}
\affiliation{Perimeter Institute for Theoretical Physics, Waterloo, ON N2L 2Y5, Canada}

\begin{abstract}\noindent
Random quantum circuits are a central concept in quantum information theory with applications ranging from demonstrations of quantum computational advantage to descriptions of scrambling in strongly-interacting systems and black holes. 
The utility of random quantum circuits in these settings stems from their ability to rapidly generate quantum pseudo-randomness.
In a seminal paper by Brand{\~a}o, Harrow, and Horodecki, it was proven that the $t$-th moment operator of local random quantum circuits on $n$ qudits with local dimension $q$ has a spectral gap of at least $\Omega(n^{-1}t^{-5-3.1/\log(q)})$, which implies that they are efficient constructions of approximate unitary designs.

As a first result, we use Knabe bounds for the spectral gaps of frustration-free Hamiltonians to show that $1D$ random quantum circuits have a spectral gap scaling as $\Omega(n^{-1})$, provided that $t$ is small compared to the local dimension: $t^2\leq O(q)$.
This implies a (nearly) linear scaling of the circuit depth in the design order $t$.
Our second result is an unconditional spectral gap bounded below by $\Omega(n^{-1}\log^{-1}(n) t^{-\alpha(q)})$ for random quantum circuits with all-to-all interactions.
This improves both the $n$ and $t$ scaling in design depth for the non-local model.
We show this by proving a recursion relation for the spectral gaps involving an auxiliary random walk. 
Lastly, we solve the smallest non-trivial case exactly and combine with numerics and Knabe bounds to improve the constants involved in the spectral gap for small values of $t$.
\end{abstract}

\maketitle

Random unitary matrices are widespread in quantum information theory, with applications in tomography, state distinguishability, cryptography, randomized benchmarking, and decoupling. Nevertheless, achieving full uniform randomness can be prohibitively expensive and requires exponential resources.
Therefore, one often resorts to "less random" probability distributions, so-called unitary $t$-designs~\cite{dankert_exact_2009,gross_evenly_2007}.
These are probability distributions that produce the same expectation values as the uniform (Haar) measure on the unitary group up to polynomials of degree $t$.

As opposed to full Haar-randomness, unitary designs can be approximately generated with polynomial resources. Specifically, it is known that random quantum circuits (RQCs), with randomly chosen two-local unitary gates, form approximate unitary designs \cite{HL08,brandao_local_2016,HM18,NHJ19,brandao2010exponential}, following a line of work studying their convergence properties \cite{ELL05,ODP07,Zni08,BV10}.
In particular, it was proven in Ref.~\cite{brandao_local_2016} that $n$-qudit random quantum circuits in a parallelized one-dimensional architecture constitute $\varepsilon$-approximate unitary $t$-designs in depth $O(nt^{11})$. 

Recently, a direct connection between higher approximate designs and circuit complexity was established in Ref.~\cite{CompGrowth19}.
This result implies that the complexity of quantum circuits of depth $T=O(nt^{11})$ grows at least as $T^{1/11}$.
In fact, a conjecture by Brown and Susskind~\cite{brown2018second,susskind2018black}, motivated by the long-time behavior of black holes in the context of the AdS/CFT correspondence, 
anticipates that the complexity of local random circuits grows linearly in time for an exponentially long time.
This would be implied by a $T=O(nt)$ scaling of the circuit depth. Finally, progress towards proving this scaling was made in Ref.~\cite{NHJ19}, using a mapping to the statistical mechanics of a lattice model to show that $O(nt)$ depth RQCs form approximate $t$-designs in the limit of large local dimensions.

In this work we first show that it suffices to choose the local dimension as $q= \Omega(t^2)$, thus independent of the system size and polynomial in $t$ to ensure that random quantum circuits of depth $O(nt\log(t))$ converge to approximate unitary $t$-designs.
This is a consequence of a bound on the \textit{spectral gap} of random quantum circuits, which we define as the difference between the highest eigenvalue of $1$ and the second highest eigenvalue $\lambda_2$ of the $t$-th moment operator.
The key tool we use is a finite-size criterion for spectral gaps, so-called Knabe bounds~\cite{knabe1988energy,GM15} as opposed to the martingale (or Nachtergaele) method~\cite{nachtergaele1996spectral} used in~\cite{brandao_local_2016}.
Combining these bounds with an approximate orthogonality result from Ref.~\cite{brandao_local_2016} we obtain a simple proof that random quantum circuits have a constant spectral gap in the regime $q=\Omega(t^2)$.

Our second result is a polynomial spectral gap for random quantum circuits with all-to-all interactions, which we will refer to as \textit{non-local} random quantum circuits.
In this model, the best known bound on the spectral gap is derived from the $1D$ result and scales like $\Omega(n^{-2}t^{-5-3.1/\log(q)})$.
We prove a recursion relation for the spectral gap involving the spectral gap of an auxiliary random walk. 
We use methods from~\cite{brandao_local_2016}, specifically a version of path coupling method by Bubley and Dyer~\cite{bubley1997path} due to Oliveira~\cite{oliveira2009convergence}, to bound the auxiliary walk and obtain a bound on the spectral gap of the non-local model of $\Omega(n^{-1}\log^{-1}(n)t^{-\alpha(q)})$, with an improvement in $n$ dependence, and where
\begin{equation}
	\alpha(q):=2.03\log^{-1}(q)\log\left(1-\left(1-\frac{1}{ 2.1q^2}\right)^{\frac14}\right)^{-1}.
\end{equation}
In particular, we have $\lim_{q\to\infty}\alpha(q)=4.06$, hence also slightly improving the $t$ dependence for the non-local model.
We do not require an application of the martingale method and believe that the auxiliary walk might be a useful tool towards the proof of a constant spectral gap.

Furthermore, we analytically and numerically improve on the spectral gaps for small values of $t$.
We prove an exact formula for the smallest non-trivial case, $n=3$ and $t=2$, for $1D$ RQCs with open boundary conditions.
We find that the second highest eigenvalue of the moment operator in this case is exactly 
\begin{equation}
\lambda_2=\frac12 +\frac{q}{2(q^2+1)}\,.
\end{equation}
Combined with Knabe bounds, this yields much smaller explicit constant for the generation of approximate $2$-designs.
Similarly, we numerically compute the local spectral gaps for small values of $q$ and $t$ to obtain improved constants in the design depths for $1D$ random quantum circuits. 
Lastly, we extend the results to random quantum circuits constructed from local orthogonal gates, and show that orthogonal random quantum circuits converge to approximate orthogonal $t$-designs, reproducing moments of the Haar measure on the orthogonal group.

\section{Preliminaries}
A central object of this paper is the moment superoperator, the $t$-fold channel of an operator $A$ with respect to a probability distribution $\nu$ on the unitary group $U(d)$, defined as
\begin{equation}
\Phi_\nu^{(t)}(A):=\int U^{\otimes t}A(U^{\dagger})^{\otimes t}\, \mathrm{d}\nu(U)\,.
\end{equation}
We denote the Haar-measure on the unitary group by $\mu_H$.

We can then use the vectorization isomorphism $\mathrm{vec}:\mathbb{C}^{D\times D}\to \mathbb{C}^{D^2}$ defined by $\mathrm{vec}(|i\rangle\langle j|)=|i\rangle\otimes |j\rangle$.
This isomorphism uniquely extends to a map from superoperators to matrices: $\mathrm{vec}(T)\mathrm{vec}(M):=\mathrm{vec}(T(M))$ for all $M\in \mathbb{C}^{D\times D}$ for a superoperator $T$.

A principal notion for us will be the spectral gap of moment operators:
\begin{equation}
g(\nu,t):=\big\|M(\nu,t)-M(\mu_H,t)\big\|_{\infty}\,,
\end{equation}
where the $t$-th moment operator of a probability distribution is defined as
\begin{equation}
M(\nu,t):=\mathrm{vec}\big(\Phi_\nu^{(t)}\big)=\int U^{\otimes t}\otimes \overline{U}{}^{\otimes t}\,\mathrm{d}\nu(U)
\end{equation} 
and $\|\bullet\|_{\infty}$ denotes the Schatten $\infty$-norm.
In particular, the spectral gap $g(\nu,t)$ can be easily amplified, where the $k$-fold convolution of $\nu$ has the property that
\begin{equation}\label{eq:gconvolution}
g(\nu^{*k},t)\leq g(\nu,t)^k\,.
\end{equation}
Upper bounds on this spectral gap can be used to imply an approximate version of unitary designs~\cite{brandao_local_2016}. We define approximate designs in two (inequivalent) ways, with a relative error and with an exponentially small additive error. As we will shortly see, the relation to the spectral gap turns out to be the same.
\begin{definition}[Approximate unitary designs] \label{def:approxdesign}
	\textcolor{white}{a}
\begin{enumerate}
\item A probability distribution $\nu$ on $U(d)$ is an $\varepsilon$-approximate unitary $t$-design if the $t$-fold channel obeys 
\begin{equation} 
\big\|\Phi_\nu^{(t)}-\Phi_{\mu_H}^{(t)}\big\|_\diamond \leq \frac{\varepsilon}{d^t}\,.
\label{eq:approxdesign1}
\end{equation}
\item A probability distribution $\nu$ on $U(d)$ a \emph{(relative)} $\varepsilon$-approximate unitary $t$-design if 
\begin{equation}
(1-\varepsilon)\Phi_\nu^{(t)} \preccurlyeq \Phi_{\mu_H}^{(t)} \preccurlyeq	(1+\varepsilon)\Phi_\nu^{(t)}\,,
\label{eq:approxdesign2}
\end{equation}
where here $A\preccurlyeq B$ if and only if $B-A$ is a completely positive map.
\end{enumerate}
\end{definition}

Combined with the above definition of an approximate unitary design, Lemma 4 in \cite{brandao_local_2016}, as well as the fact that $\|\Phi^{(t)}_\nu - \Phi^{(t)}_{\mu_H}\|_\diamond \leq d^t g(\nu,t)$, allow us to establish the following:
\begin{lemma}\label{lemma:gtodesign}
Let $\nu$ be a probability distribution on $U(d)$ such that $g(\nu,t)\leq \varepsilon/d^{2t}$. 
Then $\nu$ is an $\ep$-approximate unitary $t$-design and obeys both Eq.~\eqref{eq:approxdesign1} and Eq.~\eqref{eq:approxdesign2}.
\end{lemma}
\ni Therefore, whenever we refer to an $\ep$-approximate design in this work, we mean in both the additive and relative sense in \autoref{def:approxdesign}.

In this paper we consider the following architectures of random quantum circuits comprised of 2-local unitary gates on a system of $n$ qudits with local dimension $q$:
\begin{definition}[Random quantum circuits] \label{def:RQCs}
	\textcolor{white}{a}
	\begin{enumerate}
    \item \emph{Local ($1D$) random quantum circuits:} Let $\nu_n$ denote the probability distribution on $U((\mathbb{C}^{q})^{\otimes n})$ defined by first choosing a random pair of adjacent qudits and then applying a Haar random unitary $U_{i,i+1}$ from $U(q^2)$.
    Without further clarification we assume periodic boundary conditions (pbc), i.e. we identify the qudits $1$ and $n+1$,
    else we speak of local random circuits with open boundary conditions (obc).
    \item \emph{Brickwork random quantum circuits:} Apply first a unitary $U_{1,2}\otimes U_{3,4}\otimes...$ and then a unitary $U_{2,3}\otimes U_{4,5}\otimes...$, where all $U_{i,i+1}$ are drawn Haar-randomly. For simplicity we assume in this case an even number of qudits.  
    We denote this distribution by $\nu^{\rm bw}_n$.
    \item \emph{Non-local random quantum circuits:}
    In each step of the non-local random quantum circuit we randomly draw a pair of qudits $(i,j)$ and apply a Haar-random gate from $U(q^2)$ to this subsystem.
    We denote the corresponding measure on $U(q^n)$ by $\nu^{\mathrm{non}}_n$.
	\end{enumerate}
\end{definition}

Each of these probability distributions defines a single time step for the random quantum circuit model. We will often discuss the RQC {\it depth}.
A depth $T$ random quantum circuit will refer to the evolution after $T$ time steps in the model, namely the distribution $\nu^{*T}$.
Note that in the case of brickwork RQCs, each time step consists of two layers. 

We now mention some previous results which computed the design depth for random quantum circuits. 
As we will be interested in manipulating the local dimension to improve on previous results, we present two extremes in this regard.

\begin{theorem}[Cor.\ 6 in \cite{brandao_local_2016}]
Local random quantum circuits on $n$ qubits, $q=2$, form $\varepsilon$-approximate unitary designs if the circuit depth is
\begin{equation}
    T \geq C n \lceil \log (4t)\rceil^2 t^{9.5} (2nt + \log 1/\varepsilon)\,,
\end{equation}
where the constant is taken to be $C=4\times 10^7$.
\end{theorem}
At the other extreme, we have:
\begin{theorem}[\cite{NHJ19}]\label{thm:largeq}
Brickwork random quantum circuits on $n$ qudits, with large local dimension $q$, form $\varepsilon$-approximate unitary designs if the circuit depth is
\begin{equation}
T \geq 4nt + \log 1/\varepsilon\,,
\end{equation}
for some $q\geq q_0$ which depends on $t$ and the size of the circuit.
\end{theorem}

Part of the goal of this work is to try and close the gap between these two results. We focus on the former approach, where was observed that the circuit size $T$ required for local random quantum circuits to form an $\ep$-approximate unitary $t$-designs can be determined from the spectral gap $\Delta$ of a Hamiltonian, as described below.
Lastly, Ref.~\cite{brandao_local_2016} also showed that a lower bound on the depth needed for $1D$ random circuits to form unitary designs is $\Omega(nt/\log(nt))$, and thus the linear scaling in $n$ and $t$ cannot be further improved. 

For higher-dimensional random quantum circuits, the scaling in the number of qudits can be improved, and in Ref.~\cite{HM18} they showed that RQCs on a $D$-dimensional lattice form approximate designs in $O(n^{1/D}{\rm poly}(t))$ depth. It remains to be seen if a linear design growth holds in higher dimensions. Other (non-RQC) implementations of approximate unitary designs are also known \cite{HL09}, including some time-dependent Hamiltonian constructions \cite{Nakata16,Onorati17}. More recently, Ref.~\cite{QMhomeopathy20} took a different approach towards efficiency and proved that $O(n\, {\rm poly}(t))$ depth random Clifford circuits are approximate $t$-designs (for $t^2\leq O(n)$) with only $\tilde O(t^4)$ non-Clifford gates dispersed throughout the circuit.

We end the section by emphasizing the utility of high-degree designs. While some applications of approximate unitary designs in the literature only require control over the first few moments, higher moments are important for establishing concentration bounds \cite{LowDeviation09} and have recently been essential in proving statements about the saturation of entanglement \cite{Liu2018}, the late-time equilibration of subsystems \cite{EntFlucs20}, and the growth of quantum complexity \cite{CompGrowth19}. Specifically, Ref.~\cite{CompGrowth19} proved a linear relation between the circuit complexity of unitaries in an approximate design and the degree of the design $t$. This was established for both the standard circuit complexity of a unitary, as well as a stronger notion of complexity in terms of optimal distinguishing measurements. Consequent to this work is a relation between circuit depth and complexity growth; rigorously showing a linear design growth proves a linear growth of the quantum complexity in time.

\section{Constant spectral gap for large local dimensions from Knabe bounds}

In this section we bound spectral gaps for large local dimensions and use them to deduce the depth at which random quantum circuits form designs.
\begin{theorem}[Spectral gaps for large $q$]\label{thm:gapbound}
	Local random quantum circuits have a spectral gap that can be bounded by
	\begin{equation}
	g(\nu_n,t)\leq 1 - \frac{1}{2n}
	\end{equation}
	for all $q\geq 6t^2, n\geq 4$ and $t\geq 1$.
\end{theorem}
As explained in the preliminaries, this implies the following:
\begin{corollary}[Unitary designs for large $q$]\label{cor:RQCdesigns}
	Assume that $q\geq 6t^2$ and $n\geq 4$. Then the following statements hold:
	\begin{enumerate}
	\item  Local random quantum circuits of depth $2n(2nt\log(q)+\log(1/\varepsilon))$ are $\ep$-approximate unitary $t$-designs.
	\item Brickwork random quantum circuits of depth $18(2nt\log(q)+\log(1/\varepsilon))$ are $\ep$-approximate unitary $t$-designs.
	\end{enumerate}
\end{corollary}

Notice that for $q=\Omega(t^2)$, we have an ultimate scaling of $O(nt\log(t))$.
We can insert $O(t\log(t))$ into the main result of Ref.~\cite{CompGrowth19} to show that the complexity of the vast majority of instances has almost linear complexity growth at least until $t\sim\sqrt{q}$.
This provides further evidence for the long-time linear growth of quantum complexity.

As in Ref.~\cite[Lem.~16]{brandao_local_2016}, it proves useful to reformulate the difference in operator norm in terms of a one-dimensional local Hamiltonian:
\begin{equation}\label{eq:gHspecgap}
g(\nu_n,t)=1-\frac{\Delta(H_{n,t})}{n}\,,
\end{equation}
where
\begin{equation}
H_{n,t}=\sum_{i=1}^nP_{i,i+1}, \qquad\text{and}\qquad P_{i,i+1} := \iden - \iden_{[1,i-1]} \otimes P_H^{(2)} \otimes \iden_{[i+2,n]}\,,
\end{equation}
where we introduced the shorthand notation
\begin{equation}
P^{(m)}_H:=M(\mu_H,t)
\end{equation}
on $m$ qudits.
The local Hamiltonian $H_{n,t}$ is frustration-free, i.e. it has a ground space with eigenvalue $0$. 
In fact this ground space can be characterized as the space spanned by permutations.
Denote by $r(\pi)$ the standard representation of a permutation $\pi\in S_t$:
\begin{equation}
r(\pi)|j_1\rangle\otimes...\otimes|j_t\rangle:=|j_{\pi(1)}\rangle\otimes...\otimes|j_{\pi(t)}\rangle.
\end{equation}
Then the ground space of $H_{n,t}$ is spanned by the vectors $\ket{\psi_\pi}^{\otimes n}$, where
\begin{equation}
|\psi_{\pi}\rangle:= q^{-t/2}\mathrm{vec}(r(\pi))=(\iden\otimes r(\pi))|\Omega\rangle,\qquad |\Omega\rangle:=q^{-t/2}\sum_{i=1}^{q^{t}}|i,i\rangle\,,
\end{equation}
where $|i\rangle$ denotes an orthonormal basis of $(\mathbb{C}^{q})^{\otimes t}$.
Notice that $|\psi_{\pi}\rangle$ is normalized with respect to the Frobenius norm.
Moreover, we denote the density matrix of the states as
\begin{equation}
\psi_{\pi}:=|\psi_{\pi}\rangle\langle\psi_{\pi}|\,.
\end{equation}
To bound the gap of the Hamiltonian, we use the following finite-size criteria from Ref.~\cite{knabe1988energy,GM15}:
\begin{lemma}[Knabe bound]\label{lemma:knabe}
	Consider a frustration-free translation-invariant Hamiltonian $H_n=\sum_{i=1}^n P_{i,i+1}$ with projectors $P_{i,i+1}$.
	Define the bulk Hamiltonian $H_m^{\rm bulk}:=\sum_{i=1}^{m-1}P_{i,i+1}$.
	Let $m>2$ and $n>2m$. Then
	\begin{equation}
	\Delta(H_n)\geq \frac{5}{6}\frac{m^2+m}{m^2-4}\left(\Delta(H^{\rm bulk}_m)-\frac{6}{m(m+1)}\right)\,.
		\end{equation}
\end{lemma}
\ni In particular, we need the bound for $m=3$:
\begin{equation}\label{eq:knabem=3}
\Delta(H_n)\geq 2\left(\Delta(H^{\rm bulk}_3)-\frac{1}{2}\right).
\end{equation}
We review Knabe bounds of this type and some generalizations in \autorefapp{app:knabe}.

\vspace*{4pt}
We proceed by defining the \textit{frame operator} of the basis $\{|\psi_{\pi}\rangle\}$  as:
\begin{equation}
S:=\sum_{\pi\in S_t}\psi_{\pi}.
\end{equation}
The following lemma was proven in Ref.~\cite{brandao_local_2016}.
\begin{lemma}[Approximate orthogonality of permutations]\label{lemma:BHHinequality}
Consider the Haar-projector $P_H$ on the unitary group $U(D)$.
Assume that $D>t^2$. Then, the following bound holds:
 \begin{equation}
 \|P_H-S\|_{\infty}\leq \frac{t^2}{D}.
 \end{equation}
\end{lemma}
The proof of \autoref{lemma:BHHinequality} carries over from Ref.~\cite[Lem.~16]{brandao_local_2016} without further modifications.
\begin{proof}[Proof of \autoref{thm:gapbound}]
	Consider the probability distribution that applies a Haar-random unitary from $U(q^2)$ to a random pair of qudits $(i,i+1)$ with $1\leq i\leq m-1$.
	Denote this probability distribution by $\nu^{\rm bulk}_m$.
Then, we have the difference of moment operators
\begin{equation}
M\big(\nu^{\rm bulk}_3,t\big)-M(\mu_H,t) = \frac12\left(P^{(2)}_{H}\otimes \iden+\iden\otimes P^{(2)}_H\right)-P_H^{(3)}\,.
\end{equation}
This expression is a positive semidefinite operator.
We can thus apply \autoref{lemma:BHHinequality} and obtain
\begin{align}
\begin{split}
 \big\|M\big(\nu^{\rm bulk}_3,t\big)-M(\mu_H,t)\big\|_{\infty} &\leq \left\|\frac12\left(S^{(2)}\otimes \iden + \iden\otimes S^{(2)}\right)-S^{(3)}\right\|_{\infty}+\frac{t^2}{q^2}+\frac{t^2}{q^3}\\
& \leq \frac{1}{2} \bigg\|\sum_{\pi}\left(\psi_{\pi}^{\otimes 2}\otimes \iden + \iden \otimes \psi_{\pi}^{\otimes 2}\right)-2\psi_{\pi}^{\otimes 3}\bigg\|_\infty + \frac{t^2}{q^2}+\frac{t^2}{q^3}\,.
\end{split}
\end{align}
Consider an orthonormal basis $\{|\pi\rangle\}$ for $\spann\{|\psi_{\pi}\rangle\}$ and define the \textit{synthesis operator}
\begin{equation}
B:=\sum_{\pi\in S_t}|\pi\rangle\langle\psi_{\pi}|\,.
\end{equation}
Notice that $B^{\dagger}B=S$.
Then, we have
\begin{align}
&\big\|M\big(\nu^{\rm bulk}_3,t\big)-M(\mu_H,t)\big\|_\infty\\
&\leq \frac12 \bigg\|(\iden \otimes B^\dagger \otimes \iden) \bigg(\!\sum_{\pi}\psi_{\pi} \otimes\ketbra{\pi}\otimes \iden + \iden \otimes \ketbra{\pi}\otimes \psi_{\pi} -2\psi_{\pi}\otimes \ketbra{\pi}\otimes \psi_{\pi}\!\bigg)(\iden\otimes B\otimes \iden)\bigg\|_{\infty}\nn
&\quad +\frac{t^2}{q^2}+\frac{t^2}{q^3}\,. \nonumber
\end{align}
Hence, we can upper bound as follows
\begin{align}
\begin{split}
&\big\| M\big(\nu^{\rm bulk}_3,t\big)-M(\mu_H,t)\big\|_\infty\\
&\quad \leq \frac12 \|B^{\dagger}B\|_\infty \max_{\pi}\big\|\psi_{\pi}\otimes \iden+\iden\otimes \psi_{\pi}-2\psi_{\pi}\otimes \psi_{\pi}\big\|_{\infty}+\frac{t^2}{q^2}+\frac{t^2}{q^3}\\
&\quad \leq \frac12\left(1+\frac{t^2}{q}\right)\max_{\pi}\Big\|\psi_{\pi}\otimes \psi_{\pi}^{\perp}+\psi_{\pi}^{\perp}\otimes \psi_{\pi}\Big\|_{\infty}+\frac{t^2}{q^2}+\frac{t^2}{q^3}\\
&\quad = \frac12\left(1+\frac{t^2}{q}\right)+\frac{t^2}{q^2}+\frac{t^2}{q^3}\,.
\end{split}
\end{align}
Therefore, choosing $q\geq 6t^2$, we have
\begin{equation}
\big\|M\big(\nu^{\rm bulk}_3,t\big)-M(\mu_H,t)\big\|_{\infty}\leq \frac58.
\end{equation}
By block diagonalization, this immediately implies the operator inequality
\begin{equation}
\iden^{\otimes 3}-M\big(\nu^{\rm bulk}_3,t\big) = \iden^{\otimes 3} - \frac12 \left(P^{(2)}_{H}\otimes \iden + \iden \otimes P^{(2)}_H\right)\geq \frac38 \left(\iden-P^{(3)}_H\right)\,,
\end{equation}
and in turn
\begin{equation}
\Delta\big(H^{\rm bulk}_{3,t}\big) = \Delta\Big(\iden^{\otimes 3}-P^{(2)}_{H} \otimes \iden + \iden^{\otimes 3} - \iden \otimes P^{(2)}_H\Big) \geq \frac68\,.
\end{equation}
Plugging this bound into Eq.~\eqref{eq:knabem=3}, we end up with
\begin{equation}
\Delta(H_{n,t})\geq 2\left(\frac68-\frac12\right)=\frac12\,.
\end{equation}
\end{proof}
We can proceed to prove \autoref{cor:RQCdesigns}.
\begin{proof}[Proof of \autoref{cor:RQCdesigns}]
	
\autoref{thm:gapbound} immediately implies the first item of \autoref{cor:RQCdesigns} using Eq.~\eqref{eq:gconvolution} combined with \autoref{lemma:gtodesign}.
To obtain the second design depth for brickwork random quantum circuits we apply the generalized version~\cite{anshu2016simple} of the detectability lemma~\cite{aharonov2009detectability}:
\begin{lemma}[Detectability lemma] \label{lemma:detect}
Let $H=\sum_{i=1}^m Q_i$ be a frustration-free Hamiltonian with $\{Q_1,...,Q_m\}$ a set of orthogonal projectors.
Assume that each $Q_i$ commutes with all but $g$ of the projectors.
Then, for any state $|\psi^{\perp}\rangle$ orthogonal to the ground space of $H$,
\begin{equation}
\left\|\prod_{i=1}^{m} (\iden-Q_i)|\psi^{\perp}\rangle\right\|_2^2\leq \frac{1}{\Delta(H)/g^2+1}\,.
\end{equation}
\end{lemma}
This can be directly applied to the moment operator $M(\nu^{\rm bw}_n,t)$.
Similar to the application in \cite{haferkamp2019closing}, we obtain the following bound on the brickwork spectral gap
\begin{equation}
g(\nu^{\rm bw}_n,t)= \big\| M(\nu^{\rm bw}_n,t)-M(\mu_H,t) \big\|_{\infty}\leq \frac{1}{\sqrt{\Delta(H_{n,t})/4+1}}\leq 1-\frac{1}{18}\,.
\end{equation}
Therefore, we can again apply Eq.~\eqref{eq:gconvolution} and \autoref{lemma:gtodesign}, which completes the proof of \autoref{cor:RQCdesigns}.
\end{proof}

\section{A spectral gap for non-local random quantum circuits}
In this section we consider \textit{non-local} random quantum circuits. 
Non-local is meant in a geometric sense, where the circuit architecture is defined on a complete graph of the qudits, as described in \autoref{def:RQCs}

A polynomial spectral gap can be deduced from the $1D$ result in Ref.~\cite{brandao_local_2016} since the non-local random quantum circuits contain have overlapping support with the $1D$ circuits.
More precisely, a randomly drawn pair of qudits is nearest neighbour with probability $\sim 1/n$.
This yields a spectral gap of $\Omega(n^{-2}\mathrm{poly}^{-1}(t))$ of the moment operator for non-local circuits.

Here, we prove a recursion relation for the non-local circuits that allows us to prove a scaling of $\Omega(n^{-1}\log^{-1}(n)\mathrm{poly}^{-1}(t))$.
Our proof does not require an application of the Nachtergaele method~\cite{nachtergaele1996spectral}.
Moreover, we obtain slightly improved exponents in $t$ for large local dimensions.
More precisely, we show the following result:
\begin{theorem}[Spectral gap for non-local random quantum circuit]\label{thm:spectralnonlocal}
	Let $n\geq \max\{\lceil 2.03\log_q(t)\rceil,6000\}$, then there is a constant $c(q)$ such that 
	\begin{equation}
    \big\|M(\nu_n^{\mathrm{non}},t)-M(\mu_H,t)\big\|_{\infty}\leq 1-c(q)n^{-1}\log^{-1}(n)\log(t)t^{-\alpha(q)},
	\end{equation}
	with
    \begin{equation}\label{eq:alpha}
	\alpha(q):=2.03\log^{-1}(q)\log\left(1-\left(1-\frac{1}{ 2.1q^2}\right)^{\frac14}\right)^{-1}.
	\end{equation}
\end{theorem}
This implies the following result about unitary designs.

\begin{corollary}\label{cor:non-localdesign}
	Let $n\geq  \max\{\lceil2.03\log_q(t)\rceil,6000\}$, then there is a constant $C(q)$ such that $(\nu^{\mathrm{non}}_n)^{*T}$ is an $\varepsilon$-approximate unitary $t$-design for $T\geq C(q)n\log(n)\log^{-1}(t) t^{\alpha(q)}(2nt\log(q)+\log(1/\varepsilon))$, with $\alpha(q)$ as in Eq.~\eqref{eq:alpha}.
\end{corollary}
Most notably, $\lim_{q\to \infty}\alpha(q)=4.06$. 
\begin{remark}
The prefactor in Eq.~\eqref{eq:alpha} of $2.03$ can be pushed down arbitrarily close to $2$ by imposing the condition that $n$ is larger than some constant.
For simplicity we have chosen a specific example of this trade-off by imposing that $n\geq 6000$.
Moreover, at the expense of having a higher exponent in $n$, we could also obtain a prefactor of $2$ and hence a limiting exponent of $t^4$ for large local dimensions.
\end{remark}
To prove \autoref{cor:non-localdesign}, we show a recursion relation for the spectral gap of the moment operators $M(\nu^{\mathrm{non}}_n,t)$ to the spectral gap of an auxiliary random walk, which is in a sense anti-local. 
This technique is reminiscent of a method used by Maslen in Ref.~\cite{maslen2003eigenvalues} to compute the spectral gaps of Kac's random walk~\cite{kac1947random} on $SO(N)$.
We then combine this recursion relation with a bound on the spectral gap of the auxiliary walk in two different regimes using techniques from Ref.~\cite{brandao_local_2016} and the path coupling method of Bubley and Dyer~\cite{bubley1997path} on the unitary group~\cite{oliveira2009convergence}.
Moreover, the application of the path coupling technique is slightly simplified as we only require two steps of the auxiliary walk as opposed to $n$ steps.
Overall, the structure of the argument resembles the proof in Ref.~\cite{brandao_local_2016}:
We first solve the auxiliary spectral gap problem in the regime $n\geq O(\log(t))$ for which we obtain the desired result due to the approximate orthogonality of the permutation operators and combine this with a general bound independent of $t$ but exponential in $n$.

\begin{figure}[h]
\definecolor{bblue}{HTML}{56A5EC}
\begin{tikzpicture}[scale=0.68,thick]
\foreach \y in {0,1,2,3,4,5}{
\draw (0,\y) -- (12,\y);}
\foreach \x/\y in {1/0,3.5/1,11/3}
{\draw[fill=bblue,rounded corners] (\x-0.5,0-0.42) -- ++(0,\y+0.92) -- ++(1,0) -- ++(0,-\y-0.92);}
\foreach \x/\y in {1/3,3.5/2,11/0}
{\draw[fill=bblue,rounded corners] (\x-0.5,5+0.42) -- ++(0,-\y-0.92) -- ++(1,0) -- ++(0,\y+0.92);}
\foreach \x/\y in {6/4}
{\draw[fill=bblue,rounded corners] (\x-0.5,5+0.42) -- ++(0,-\y-0.92) -- ++(1,0) -- ++(0,\y+0.92) -- cycle;}
\foreach \x/\y in {8.5/4}
{\draw[fill=bblue,rounded corners] (\x-0.5,0-0.42) -- ++(0,\y+0.92) -- ++(1,0) -- ++(0,-\y-0.92) -- cycle;}
\foreach[count=\i] \x/\y in {1/3.5,3.5/4,6/3,8.5/2,11/1.5}
{\node at (\x+0.03,\y) {$U_\i$};}
\end{tikzpicture}
\caption{An instance of the auxiliary walk described by $\nu^{\mathrm{aux}}_n$ for $n=6$.}
\end{figure}
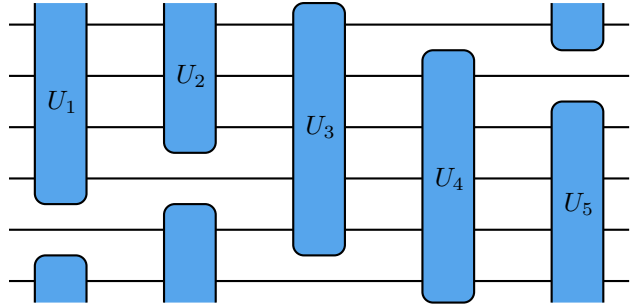

We believe that the auxiliary walk is more approachable than regular random quantum circuits and it might be a useful tool towards further improvements on the spectral gap.

We start by describing the auxiliary walk:
\begin{definition}[Auxiliary walk]
	In each step of the walk draw a random qudit $i$ and apply a Haar random unitary from $U(q^{n-1})$ to the subsystem consisting of all but the $i$-the qudit.
	We denote the corresponding probability measure on $U(q^n)$ by $\nu^{\mathrm{aux}}_n$.
\end{definition}

Moreover, we will use the following notations for the spectral gaps
\begin{equation}
\Delta_n:=\big\|M(\nu_n^{\mathrm{non}},t)-M(\mu_H,t)\big\|_{\infty},\qquad\quad \gamma_n:=\big\|M(\nu_n^{\mathrm{aux}},t)-M(\mu_H,t)\big\|_{\infty}.
\end{equation}
The key to our approach is the following recursion relation:
\begin{lemma}[Recursion relation for non-local gap] \label{lemma:recursionauxiliarywalk}
	For all $n>2$ it holds that 
	\begin{equation}
	\Delta_n\leq \gamma_n+\Delta_{n-1}(1-\gamma_n).
	\end{equation}
\end{lemma}
\begin{proof}
Let $S$ be a subset of qudits. 
We denote with $\mu_{S,H}$ the Haar measure on $U(q^{|S|})$ acting on the subsystem consisting of the qudits in $S$.
We further denote 
\begin{equation}
P_{ij}:=M(\mu_{ij,H},t)\qquad\quad Q_i:=M(\mu_{[1,i-1]\cup [i+1,n],H},t).
\end{equation}
We use repeatedly the characterization
\begin{equation}
\big\|M(\nu,t)-M(\mu_H,t)\big\|_{\infty}=\max_{\substack{\psi\in\spann\{|\psi_{\pi}\rangle,\pi\in S_t\}^{\perp}\\ \|\psi\|_2=1}}\langle\psi|M(\nu,t)|\psi\rangle.
\end{equation}
In the following let $|\psi\rangle$ denote a state in $\spann\{|\psi_{\pi}\rangle,\pi\in S_t\}^{\perp}$.
We obtain
\begin{align}
\begin{split}
\Bigg\langle \psi\Bigg|\sum_{1\leq i<j\leq n-1}P_{ij}\Bigg|\psi\Bigg\rangle&
 =\frac{(n-1)(n-2)}{2}\left\langle \psi|Q_n|\psi\right\rangle+\sum_{1\leq i<j\leq n-1}\langle\psi|P_{ij}Q_n^{\perp}|\psi\rangle\\
&\leq \frac{(n-1)(n-2)}{2}\left(\langle \psi|Q_n|\psi\rangle+\Delta_{n-1}\langle \psi|Q_n^{\perp}|\psi\rangle\right).
\end{split}
\end{align}
Here, we cut out the $n$-th qudit. 
The same calculation works for every qudit $i$. 
Summing over all the resulting inequalities yields
\begin{align}
\begin{split}
(n-2)\Bigg\langle \psi\Bigg|\sum_{1\leq i<j\leq n}P_{ij}\Bigg|\psi\Bigg\rangle&\leq \frac{(n-1)(n-2)}{2}\left((1-\Delta_{n-1})\sum_i\langle \psi|Q_i|\psi\rangle+n\Delta_{n-1}\right)\\
&\leq \frac{(n-1)(n-2)}{2}\left((1-\Delta_{n-1})n\gamma_n+n\Delta_{n-1}\right)\\
&\leq \frac{n(n-1)(n-2)}{2}\left((1-\Delta_{n-1})\gamma_n+\Delta_{n-1}\right).
\end{split}
\end{align}
Dividing the inequality by $\frac{n(n-1)(n-2)}{2}$ yields the result.
\end{proof}
In order to apply this recursion relation we prove two bounds on $\gamma_n$ that we will used in the large $n$ and small $n$ regimes, respectively.
The first bound only holds for $n$ large compared to $\log(t)$ but the second bound is independent of $t$ and holds for all $n$. 
\begin{lemma}[Gap bound for large $n$]\label{lemma:spectralgapauxiliarylargen}
	For $n-2\log(n)\geq 2\log(t)$ we have:
	\begin{equation}
	\gamma_n\leq \frac{1}{n}+2\frac{nt}{q^{n/2-1}}.
	\end{equation}
\end{lemma}
\begin{proof}
	We compute:
	\begin{align}
	\begin{split}
	\gamma^2_n&=\left\|\frac{1}{n}\sum_{i=1}^nQ_i-P^{(n)}_H\right\|^2_{\infty}\\
	&=\Bigg\|\frac{1}{n^2}\sum_{i,j=1}^nQ_iQ_j-P^{(n)}_H\Bigg\|_{\infty}\\
	&\leq \frac{1}{n}\left\|\frac{1}{n}\sum_{i=1}^nQ_i-P^{(n)}_H\right\|_{\infty}+\frac{1}{n^2}\sum_{i\neq j}\left\|Q_iQ_j-P^{(n)}_H\right\|_{\infty}\\
	&\leq \frac{1}{n}\left\|\frac{1}{n}\sum_{i=1}^nQ_i-P^{(n)}_H\right\|_{\infty}+\left\|Q_1Q_n-P^{(n)}_H\right\|_{\infty}.
	\end{split}
	\end{align}
   Using \autoref{lemma:BHHinequality}, we start to bound the second term.
   In the following, $S_i$ denotes the frame operator on $n-1$ qudits acting on the subsystem of all qudits except $i$ (analogous to the definition of $Q_i$).
   \begin{align}
   \begin{split}
  \left\|Q_1 Q_n - P^{(n)}_H\right\|^2_{\infty}&=\left\|Q_1 Q_n Q_1 - P^{(n)}_H \right\|_{\infty}\\
  &=\left\|Q_1 Q_n Q_1 - Q_1 P^{(n)}_H Q_1\right\|_{\infty}\\
  &\leq \big\|Q_1 S_n Q_1 - Q_1 S Q_1\big\|_{\infty}+2\frac{t^2}{q^{n-1}}\\
  &=\left\|\sum_{\pi\in S_t}Q_1\psi_{\pi}^{\otimes n-1}\otimes \psi_{\pi}^{\perp}Q_1\right\|_{\infty}+2\frac{t^2}{q^{n-1}}\\
  &=\left\| \sum_{\pi\in S_t}\psi_{\pi}\otimes \left(P_H^{(n-1)}\psi_{\pi}^{\otimes n-2}\otimes \psi_{\pi}^{\perp}P_H^{(n-1)}\right)\right\|_{\infty}+2\frac{t^2}{q^{n-1}}.
\end{split}
   \end{align}
   Notice that the argument of $\|\bullet\|_{\infty}$ is a sum of positive operators:
   Indeed, $\psi_{\pi}^{\otimes n-2}\otimes \psi_{\pi}^{\perp}$ is an orthonormal projector and therefore
   \begin{multline}
   P_H^{(n-1)}\psi_{\pi}^{\otimes n-2}\otimes \psi_{\pi}^{\perp}P_H^{(n-1)}\\
   =P_H^{(n-1)}(\psi_{\pi}^{\otimes n-2}\otimes \psi_{\pi}^{\perp})^2P_H^{(n-1)}
   =P_H^{(n-1)}\psi_{\pi}^{\otimes n-2}\otimes \psi_{\pi}^{\perp}(P_H^{(n-1)}\psi_{\pi}^{\otimes n-2}\otimes \psi_{\pi}^{\perp})^{\dagger}.
   \end{multline}
   Hence, we have the operator inequality
   \begin{equation}
   \sum_{\pi\in S_t}\psi_{\pi}\otimes \left(P_H^{(n-1)}\psi_{\pi}^{\otimes n-2}\otimes \psi_{\pi}^{\perp}P_H^{(n-1)}\right)\leq \iden\otimes\sum_{\pi\in S_t} P_H^{(n-1)}\psi_{\pi}^{\otimes n-2}\otimes \psi_{\pi}^{\perp}P_H^{(n-1)}.
   \end{equation}
   In particular, for the largest eigenvalue we have 
   \begin{align}
   \begin{split}
   \|Q_1Q_n-P_H\|^2_{\infty}&\leq \left\|\sum_{\pi\in S_t} P_H^{(n-1)}\psi_{\pi}^{\otimes n-2}\otimes \psi_{\pi}^{\perp}P_H^{(n-1)}\right\|_{\infty}+2\frac{t^2}{q^{n-1}}\\
   &= \left\| P_H^{(n-1)}(S^{(n-2)}\otimes \iden) P_H^{(n-1)}-P_H^{(n-1)}S^{(n-1)}P_H^{(n-1)}\right\|_{\infty}+2\frac{t^2}{q^{n-1}}\\
   &\leq \left\| P_H^{(n-1)}(P_H^{(n-2)}\otimes \iden) P_H^{(n-1)}-P_H^{(n-1)}P_H^{(n-1)}P_H^{(n-1)}\right\|_{\infty}+4\frac{t^2}{q^{n-2}}\\
   &=4\frac{t^2}{q^{n-2}},
\end{split}
   \end{align}
   where we have again used \autoref{lemma:BHHinequality} in the second inequality.
   Combined we have the inequality 
   \begin{equation}
   \gamma_n^2\leq \frac{1}{n}\gamma_n+2\frac{t}{q^{n/2-1}}.
   \end{equation}
   If $\gamma_n< 1/n$, we are already done.
   If $\gamma_n\geq 1/n$, we obtain
   \begin{equation}
   \gamma_n\leq \frac{1}{n}+2\frac{tn}{q^{n/2-1}},
   \end{equation}
   which is the claimed bound.
\end{proof}
\begin{lemma}[General gap bound]\label{lemma:auxgap2}
	For all $n\geq 3$ we have 
	\begin{equation}
	\gamma_n\leq \left(1-\frac{1}{q^2}+\frac{2}{n}\right)^{\frac14}.
	\end{equation}
\end{lemma}
    The proof is based on a bound on the Wasserstein distance between the auxilliary random walk and the Haar-measure. 
	We will use the path coupling method for the unitary group as developed in~\cite{oliveira2009convergence}.
	In fact, the application directly generalizes the application in~\cite{brandao_local_2016}.
	Here, we only need to couple two steps of the random walk which simplifies the argument.
	
	For probability measures $\nu_1,\nu_2$ we call $(X,Y)$ a coupling if $X$ and $Y$ have marginal measures $\nu_1$ and $\nu_2$.
	The $L^p$-Wasserstein distance with respect to a metric $d$ is
	\begin{equation}
	W_{d,p}(\nu_1,\nu_2):=\inf\{\mathbb{E}[d(X,Y)^{p}]^{1/p}: (X,Y)\;\text{is a coupling for}\;\nu_1,\nu_2\}.
	\end{equation}
	We prove the following bound
	\begin{lemma}
		For every integer $k\geq 1$ we have:
	\begin{equation}
	W_{\mathrm{Rie},2}((\nu^{\mathrm{aux}}_n)^{*2k},\mu_{H})\leq \left(1-\frac{1}{q^2}+\frac{2}{n}\right)^{k/2}\sqrt{2}q^{3/2}.
	\end{equation}
	\end{lemma}
\begin{proof}
	The proof is a straightforward generalization of to the one in Ref.~\cite[Lem.~25]{brandao_local_2016}.
	In the following we show the parts of the argument that need to be adjusted.
	
	We consider two steps of the random walk $\nu^{\mathrm{aux}}_n$.
	In order to apply the path coupling method, we need to show that 
	\begin{equation}
     \limsup_{\varepsilon\to 0}\sup_{X,Y}\left\{\frac{W_{\mathrm{Rie},2}((\nu_n^{\mathrm{aux}})^{*2}*\delta_X,(\nu_n^{\mathrm{aux}})^{*2}*\delta_Y)}{d_{\mathrm{Rie}}(X,Y)}:d_{\mathrm{Rie}}(X,Y)\leq \varepsilon\right\}\leq \left(1-\frac{1}{q^2}+\frac{2}{n}\right)^{\frac12}
	\end{equation}
	and apply~\cite[Lem.~24]{brandao_local_2016}.
	Instead of the Riemannian distance, we consider Frobenius distance and then use that they are the same up to first order for small points~\cite{brandao_local_2016}.
	
	The second step of the random walk applied to the fixed unitary $X$ on $U(q^n)$ yields
    \begin{equation}
    X\to \left\{\tilde{U}_{[1,i-1]\cup [i+1,n]}U_{[1,j-1]\cup [j+1,n]}X\right\}_{i,j}\,,
    \end{equation}
    each with probability $1/n^2$.
    The same transformation is undergone by $Y$.
    We introduce the following transformation
    \begin{equation}
    X'\to\left\{\tilde{U}_{[1,i-1]\cup [i+1,n]}V^{i,j}_{[1,i-1]\cup [i+1,n]}U_{[1,j-1]\cup [j+1,n]}X\right\}_{i,j},
    \end{equation}
    where $V^{i,j}_{[1,i-1]\cup [i+1,n]}$ can depedend on $U_{[1,j-1]\cup [j+1,n]}$ and $V^{i,i}=\iden$.
    $Y$ is left invariant under the transformation.
    $X', Y'$ is a random coupling for $((\nu_n^{\mathrm{aux}})^{*2}*\delta_X,(\nu_n^{\mathrm{aux}})^{*2}*\delta_Y)$.
    
   We then bound
   \begin{align}
   \begin{split}
   &\mathbb{E}\big[\|X'-Y'\|_2^2\big]\\
   &=\frac{1}{n^2}\sum_{i,j}\big\|\tilde{U}_{[1,i-1]\cup [i+1,n]}V^{i,j}_{[1,i-1]\cup [i+1,n]}U_{[1,j-1]\cup [j+1,n]}X-\tilde{U}_{[1,i-1]\cup [i+1,n]}U_{[1,j-1]\cup [j+1,n]}Y\big\|_2^2.
   \end{split}
   \end{align}
   W.l.o.g. it suffices to bound the special case $i=1$ and $j=n$:
    \begin{equation}
    \mathbb{E}\bigg[\min_{V^{1,n}_{[2,n]}}\left\|V^{1,n}_{[2,n]}U_{[1,n-1]}X-U_{[1,n-1]}Y\right\|_2^2\bigg] = 2\left(\tr(\iden)-\mathbb{E}\left\|\Tr_1(U_{[1,n-1]}XY^{\dagger}U^{\dagger}_{[1,n-1]}\right\|_1\right).
    \end{equation}
    With 
    \begin{equation}
    R:=XY^{\dagger}=e^{\ii\varepsilon H}=\iden+\ii\varepsilon-\frac{\varepsilon^2}{2}H^2+O(\varepsilon^3)
    \end{equation}
    we obtain as in~\cite{brandao_local_2016}
    \begin{equation}
    \left\|\Tr_1\left(U_{[1,n-1]}RU^{\dagger}_{[1,n-1]}\right)\right\|_1=\Tr(\iden)+\frac{\varepsilon^2}{2}\frac{1}{q}\left(\Tr\left(\Tr_1(U_{[1,n-1]}RU^{\dagger}_{[1,n-1]}\right)^2\right)-\frac{\varepsilon^2}{2}\Tr(H^2)+O(\varepsilon^3).
    \end{equation}
    This yields
    \begin{multline}
    \mathbb{E}\bigg[\min_{V^{1,n}_{[2,n]}}\left\|V^{1,n}_{[2,n]}U_{[1,n-1]}X-U_{[1,n-1]}Y\right\|_2^2\bigg]\\
    =\varepsilon^2\left(\Tr(H^2)-\frac{1}{q}\mathbb{E}\left(\Tr\left(\Tr_1\left(U_{[1,n-1]}RU^{\dagger}_{[1,n-1]}\right)^2\right)\right)\right)+O(\varepsilon^3).
    \end{multline}
    It can be shown that~\cite{brandao_local_2016}:
    \begin{multline}
   \mathbb{E} \Tr\left(\Tr_1\left(U_{[1,n-1]}RU^{\dagger}_{[1,n-1]}\right)^2\right)\\
   =\mathbb{E}\Tr\left(H\otimes H\left(U^{\dagger}_{[1,n-1]}\otimes U^{\dagger}_{[1,n-1]}\right)\left(\mathbb{F}_{[2,n],[2,n]}\otimes \iden_{1,1}\right)U_{[1,n-1]}\otimes U_{[1,n-1]}\right).
    \end{multline}
    Using~\cite[Lem.~IV.3]{abeyesinghe2009mother}, we have:
    \begin{align}
    \begin{split}
    U^{\dagger}_{[1,n-1]}\otimes U^{\dagger}_{[1,n-1]}&(\mathbb{F}_{[2,n],[2,n]}\otimes \iden_{1,1})U_{[1,n-1]}\otimes U_{[1,n-1]}\\
    &=\frac{q+q^{n-1}}{q^n+1}\frac12(\iden+\mathbb{F})+\frac{q-q^{n-1}}{q^n-1}\frac12(\iden-\mathbb{F})\\
    &=\frac12\left(\frac{q+q^{n-1}}{q^n+1}+\frac{q-q^{n-1}}{q^n-1}\right)\iden+\frac12\left(\frac{q+q^{n-1}}{q^n+1}-\frac{q-q^{n-1}}{q^n-1}\right)\mathbb{F}\\
    &=\left(\frac{q-q^{-1}}{q^n-q^{-n}}\right)\iden+\left(\frac{q^{2n-1}-q}{q^{2n}-1}\right)\mathbb{F}.
    \end{split}
    \end{align}
    This yields
    \begin{equation}
    \mathbb{E}\Big[\Tr\Big(\Tr_1(U_{[1,n-1]}RU^{\dagger}_{[1,n-1]})^2\Big)\Big]\geq \frac{q^{2n-1}-q}{q^{2n}-1}\Tr(H^2)\geq \left(\frac{1}{q}-q^{1-2n}\right)\Tr(H^2).
    \end{equation}
    This implies
    \begin{align}
    \begin{split}
    \mathbb{E}\big[\|X'-Y'\|_2^2\big]&\leq \left(1-\left(\frac{1}{q^2}+q^{1-2n}\right)\left(1-\frac{1}{n}\right)\right)\|X-Y\|^2_2\\
    &\leq \left(1-\frac{1}{q^2}+\frac{2}{n}\right)\|X-Y\|^2_2,
    \end{split}
    \end{align}
    where we used that $\|X-Y\|^2_2=\varepsilon^2 \Tr(H^2)+O(\varepsilon^3)$.
    The result follows as in Ref.~\cite[Lem.~25]{brandao_local_2016}, mutatis mutandis.
\end{proof}

\begin{proof}[Proof of \autoref{lemma:spectralgapauxiliarylargen}]
The bound on the spectral gap now follows by the following inequality for all probability measures $\nu$:
	\begin{equation}
	\big\|M(\nu,t)-M(\mu_H,t)\big\|_{\infty}\leq 2t W_{\mathrm{Rie},2}(\nu,\mu_H).
	\end{equation}
	This is proven in Ref.~\cite{brandao_local_2016} using the Kantorovich duality for the Wasserstein distance.
Therefore, we have 
\begin{equation}
\big\|M(\nu^{\mathrm{aux}},t)-M(\mu_H,t)\big\|^{2k}_{\infty}\leq 2t\left(1-\frac{1}{q^2}+\frac{2}{n}\right)^{k/2}\sqrt{2}q^{3/2}.
\end{equation}
The result follows from taking the $k$-th square root and the limit $k\to\infty$ on both sides of the inequality.
\end{proof}
Finally, we can prove \autoref{thm:spectralnonlocal} by evaluating the recursion relation with the bounds in \autoref{lemma:spectralgapauxiliarylargen} and \autoref{lemma:auxgap2}.
\begin{proof}[Proof of \autoref{thm:spectralnonlocal}]
	We have
	\begin{equation}
	\gamma_n\leq \frac{1}{n} + 2\frac{tn}{q^{n/2-1}}.
	\end{equation}
	by \autoref{thm:gapbound} for all $n-\log_q(n)\geq 2\log_q(t)$.
	The latter condition is always true if $(n\geq 2.03\log_q(t))\vee (n\geq 6000)$.
	By induction we can solve the recursion relation in Lemma~\ref{lemma:recursionauxiliarywalk} and obtain that
	\begin{equation}\label{eq:recursionsolution2}
	\Delta_n\leq 1-\frac{c'}{n\log(n)}
	\end{equation}
	for $n-5\log_q(n)\geq 2\log_q(t)$ where $c'$ can be fixed by the induction beginning. 
	Indeed, assume that Eq.~\eqref{eq:recursionsolution2} holds up to some $n\geq n\geq 2\log(t)+4\log(n+1)+2\log_q(2)+1$ and $n\geq 10$.
	Then, using \autoref{lemma:recursionauxiliarywalk}, we can compute
	\begin{align}
	\begin{split}
	\Delta_{n+1}&\leq \gamma_{n+1}+\Delta_{n}(1-\gamma_{n+1})\\
	&\leq \frac{1}{n+1}+2\frac{t(n+1)}{q^{(n+1)/2-1}}+\left( 1-\frac{c'}{n\log(n)}\right)\left(1-\frac{1}{n+1}-2\frac{t(n+1)}{q^{(n+1)/2-1}}\right)\\
	&=1-c'\left(\frac{1}{(n+1)\log(n)}-2\frac{t(n+1)}{n\log(n)q^{(n+1)/2-1}}\right)\\
	&\leq 1-\frac{c'}{(n+1)\log(n+1)}-c'\left(\frac{\log(1+1/n)}{(n+1)\log(n)\log(n+1)}-2\frac{t(n+1)}{n\log(n)q^{(n+1)/2-1}}\right).
	\end{split}
	\end{align}
	For the induction to be completed, we only need to show that the third summand is negative. 
	We use that $\log(1+x)>x/2$ for $x\leq 0.1$:
	\begin{align}
	\begin{split}
	\frac{\log(1+1/n)}{(n+1)\log(n)\log(n+1)}&\geq	\frac{1}{2n(n+1)\log(n)}\\
		&\geq 2\frac{t(n+1)}{n\log(n)q^{(n+1)/2-1}}.
	\end{split}
	\end{align}
	By taking the logarithm it is easy to see that the last inequality is equivalent to $n\geq 2\log(t)+4\log_q(n+1)+2\log_q(2)+1$.
	We can easily find $c'$ from the condition
	\begin{equation}\label{eq:deltan0}
	\Delta_{n_0}= 1-\frac{c'}{n_0\log_q(n_0)}\implies c'= (1-\Delta_{n_0})n_0\log_q(n_0).
	\end{equation}
	In particular, we choose $ n_0=\max\{\lceil 2.03\log_q(t)\rceil,6000\}$.
	This is to ensure that $n\geq 2\log_q(t)+4\log_q(n+1)+2\log_q(2)+0.5$:
	 For $n\geq 6000$ we have $0.01n\geq 4\log_q(n+1)+2\log_q(2)+0.5$ which leaves us with $0.99n\geq 2\log_q(t)\impliedby 2.03\log_q(t)$.
	
    In the regime $n\leq O(\log(t))$, we can apply \autoref{lemma:auxgap2}.
    In order to do this, observe that with complete induction it is easy to show that any function $f(n)$ satisfying the recursion relation
	\begin{equation}
	f(n)\leq a+(1-a)f(n-1) \quad \forall n\geq n_0
	\end{equation}
	with $a<1$ also satisfies that there is a constant $c$ such that 
	\begin{equation}\label{eq:recursionsolution1}
	f(n)\leq 1-c(1-a)^{n} \quad \forall n\geq n_0.
	\end{equation}
	Next we choose $n\geq 4q^2$. Then, we have from \autoref{lemma:auxgap2}:
	\begin{equation}
	\gamma_n\leq \left(1-\frac{1}{2q^2}\right)^{\frac14}.
	\end{equation}
	Combined we have
	\begin{equation}
	\Delta_{n_0}\leq 1-c(q)\left(1-\left(1-\frac{1}{2q^2}\right)^{\frac14}\right)^{n_0},
	\end{equation}
	Plugging this into Eq.~\eqref{eq:deltan0} completes the proof.
\end{proof}

\section{Exact solution for the case \ensuremath{t=2} and \ensuremath{n=3}}
In this section we solve the simplest non-trivial case $n=3$ and $t=2$ with open boundary conditions as a function of $q$. 
Via Knabe bounds this yields strong bounds for $t=2$.
More precisely, we prove the following formula:
\begin{theorem}\label{theorem:exactt=2}
The spectral gap of the second moment operator for $\nu^{\rm bulk}_3$ is
\begin{equation}\label{eq:exactt=2}
\left\|M\big(\nu_3^{\mathrm{bulk}},2\big)-M(\mu_H,2)\right\|_{\infty}=\frac12+\frac{q}{2(q^2+1)}.
\end{equation}
\end{theorem}
It follows that $\Delta(H^{\rm bulk}_{n=3,t=2}) = 1-q/(q^2+1)$, which is in agreement with the $q=2$ gap computed in Ref.~\cite{brandao2010exponential}. Then, via \autoref{lemma:knabe}, this implies
\begin{equation}
	\Delta(H_{n,2})\geq 1-\frac{2q}{q^2+1}.
\end{equation}
which then yields the following corollary for convergence to approximate 2-designs in very short depth:
\begin{corollary}\label{cor:exactgap}
    Local random quantum circuits on $n$ qudits of local dimension $q$ are $\ep$-approximate unitary $2$-designs if the circuit depth is
	\begin{equation}
	T\geq n\left(1-\frac{2q}{q^2+1}\right)^{-1}\left(4n+\log_q\left(1/\varepsilon\right)\right).
	\end{equation}
\end{corollary}

\begin{proof}[Proof of \autoref{theorem:exactt=2}]
	As was observed in Ref.~\cite{brandao2010exponential}, we have
	\begin{equation}
	\big\|M\big(\nu_3^{\mathrm{bulk}},2\big)-M(\mu_H,2)\big\|_{\infty} = \frac12 \left\|\left(P^{(1)}_H\otimes P^{(2)}_H-P^{(3)}_H\right) + \left(P^{(2)}_{H}\otimes P_H^{(1)}-P^{(3)}_H\right)\right\|_{\infty}.
	\end{equation}
	Hence, we need to consider 
	\begin{align}
	\begin{split}
	\mathrm{im}(P^{(1)}_H\otimes P^{(2)}_H)&=\spann \{\mathrm{vec}~\iden,\mathrm{vec}~\mathbb{F}\}\otimes \spann \{\mathrm{vec}~\iden^{\otimes 2},\mathrm{vec}~\mathbb{F}^{\otimes 2}\}\\
	&\stackrel{\sim}{=}\spann \{P_+^{(1)}, P_-^{(1)}\}\otimes \spann \{P_+^{(2)}, P_-^{(2)}\},
	\end{split}
	\end{align}
	where
	\begin{equation}
	P^{(m)}_{\pm}:=\frac12(\iden^{\otimes m}\pm \mathbb{F}^{\otimes m})
	\end{equation}
	are the projectors onto the symmetric and antisymmetric subspace, respectively.
	We can find the orthogonal complement of $\spann \{\iden^{\otimes 3},\mathbb{F}^{\otimes 3}\}$ by imposing the necessary and sufficient conditions 
	\begin{equation}
	\tr(V\iden)=0,\qquad \text{and} \qquad \tr(V\mathbb{F})=0
	\end{equation}
	for a general $V\in\spann \{P_+^{(1)}, P_-^{(1)}\}\otimes \spann \{P_+^{(2)}, P_-^{(2)}\}$.
	It can be easily checked that the following yields an orthonormal basis for the orthogonal complement:
	\begin{align}
	\begin{split}
	A_{12}&:=\left(\frac{2}{q^3(q^3+1)}\right)^{\frac12}\left(\left(\frac{(q-1)(q^2-1)}{(q+1)(q^2+1)}\right)^{\frac12}P^{(1)}_+\otimes P^{(2)}_+-\left(\frac{(q+1)(q^2+1)}{(q-1)(q^2-1)}\right)^{\frac12}P^{(1)}_-\otimes P^{(2)}_-\right),\\
	B_{12}&:=\left(\frac{2}{q^3(q^3-1)}\right)^{\frac12}\left(\left(\frac{(q-1)(q^2+1)}{(q+1)(q^2-1)}\right)^{\frac12}P^{(1)}_+\otimes P^{(2)}_- -\left(\frac{(q+1)(q^2-1)}{(q-1)(q^2+1)}\right)^{\frac12}P^{(1)}_-\otimes P^{(2)}_+\right).
	\end{split}
	\end{align}
Analogously, we find an orthonormal basis $\{A_{21}, B_{21}\}$ for $\spann \{P_+^{(2)}, P_-^{(2)}\}\otimes \spann \{P_+^{(1)}, P_-^{(1)}\}$.
In particular, we have 
\begin{equation}\label{eq:analyticoverlapAB}
\tr[A_{12}B_{12}]=\tr[B_{21}A_{21}]=0\,.
\end{equation}
Moreover, using the general characterizations
\begin{equation}\label{eq:decompositionsymmetricprojector}
P_{+}^{(m+n)}=P_{+}^{(m)}\otimes P_{+}^{(n)}+P_{-}^{(m)}\otimes P_{-}^{(n)},\qquad
 P_{-}^{(m+n)}=P_{+}^{(m)}\otimes P_{-}^{(n)}+P_{-}^{(m)}\otimes P_{+}^{(n)}
\end{equation}
repeatedly, we also find that
\begin{equation}
\tr[B_{12}A_{21}]=\tr[A_{12}B_{21}]=0\,.
\end{equation}
Again using Eq.~\eqref{eq:decompositionsymmetricprojector}, we can compute
\begin{align}
\begin{split}
\tr[A_{12}A_{21}]=\frac{2}{q^3(q^3+1)}
&\left(\frac{(q-1)(q^2-1)}{(q+1)(q^2+1)}\tr\left[P_+^{(1)}\right]^3-2\tr\left[P_-^{(1)}\right]^2\tr\left[P_+^{(1)}\right]\right.\\
&\quad\left.+\frac{(q+1)(q^2+1)}{(q-1)(q^2-1)}\tr\left[P_-^{(1)}\right]^2\tr\left[P_+^{(1)}\right]\right)\,.
\end{split}
\end{align}
Inserting $\tr[P_+^{(1)}]=q(q+1)/2$ and $\tr[P_-^{(1)}]=q(q-1)/2$ gives us
\begin{equation}
\label{eq:analyticoverlapA}
\tr[A_{12}A_{21}]
=\frac{q}{q^2+1}\,.
\end{equation}
Proceeding similarly for the overlap of the $B$ operators, we find
\begin{equation}\label{eq:analyticoverlapB}
\tr[B_{12}B_{21}]=-\frac{q}{q^2+1}\,.
\end{equation}
We can now use these overlaps to compute the eigenvalues of the rank $4$ matrix
\begin{multline}
\left(P^{(1)}_H\otimes P^{(2)}_H-P^{(3)}_H\right)+\left(P^{(2)}_{H}\otimes P_H^{(1)}-P^{(3)}_H\right)
\\=\left(\mathrm{vec}A_{12}(\mathrm{vec}A_{12})^{\dagger}+\mathrm{vec}B_{12}(\mathrm{vec}B_{12})^{\dagger}\right)
+\left(\mathrm{vec}A_{21}(\mathrm{vec}A_{21})^{\dagger}+\mathrm{vec}B_{21}(\mathrm{vec}B_{21})^{\dagger}\right),
\end{multline}
by applying this operator to a general state
\begin{equation}
|\phi\rangle=a_{12}\mathrm{vec} A_{12}+b_{12}\mathrm{vec} B_{12}+a_{21}\mathrm{vec} A_{21}+b_{21}\mathrm{vec} B_{21}\,.
\end{equation}
From a comparison of the coefficients for the eigenvalue equation 
\begin{equation}
\left(\left(P^{(1)}_H\otimes P^{(2)}_H-P^{(3)}_H\right)+\left(P^{(2)}_{H}\otimes P_H^{(1)}-P^{(3)}_H\right)\right)|\phi\rangle=\lambda |\phi\rangle 
\end{equation}
we obtain the following system of equations:
\begin{align}
a_{12}+\tr[A_{12}A_{21}]a_{21}&=\lambda a_{12}\label{eq:eigenvalueI}\\
a_{21}+\tr[A_{12}A_{21}]a_{12}&=\lambda a_{21}\label{eq:eigenvalueII}\\
b_{12}+\tr[B_{12}B_{21}]b_{21}&=\lambda b_{12}\label{eq:eigenvalueIII}\\
b_{21}+\tr[B_{21}B_{12}]b_{12}&=\lambda b_{21}\label{eq:eigenvalueIV}\,.
\end{align}
Combining Eqs.~\eqref{eq:eigenvalueI} and \eqref{eq:eigenvalueII}, we obtain
\begin{equation}
\frac{a_{21}}{a_{12}}=\frac{a_{12}}{a_{21}}=\frac{(\lambda-1)}{\tr[A_{12}A_{21}]}\,.
\end{equation}
As $\lambda$ and $\tr[A_{12}A_{21}]$ are real, this implies
\begin{equation}
\frac{(\lambda-1)}{\tr[A_{12}A_{21}]}=\pm 1\,.
\end{equation}
The same calculation can be done for Eqs.~\eqref{eq:eigenvalueIII} and \eqref{eq:eigenvalueIV}.
Finally, this leaves us with the following four eigenvalues
\begin{equation}\label{eq:alleigenvalues}
\lambda_{a,\pm}=1\pm\tr[A_{12}A_{21}]\,,\qquad\quad \lambda_{b,\pm}=1\pm \tr[B_{12}B_{21}]\,.
\end{equation}
Combined with Eqs.~\eqref{eq:analyticoverlapA} and \eqref{eq:analyticoverlapB}, this completes the proof of \autoref{theorem:exactt=2}.
\end{proof}
We have seen in the proof of \autoref{cor:RQCdesigns} that the eigenvalues of the moment operator become $1$, $\frac12$ and $0$ for large $q$.
This is consistent with Eq.~\eqref{eq:alleigenvalues}.

An obvious question is whether this calculation can be generalized for higher $t$.
Already for $t=3$ this gets complicated by the fact that the projectors onto irreducible representations do not in general span the full group algebra of $S_t$.
In fact, they span the center of this algebra, which coincides with the full algebra only for $t=2$.

\section{Improved constants from numerical results}

Having bounded the spectral gaps for local random quantum circuits in the case of large local dimension, and further explicitly computing the gaps for the second moment, we now turn to a numerical approach. The goal is to provide improved constants for the RQC design depth for a number of different random circuit architectures for the first few moments. As we discussed in the introduction, unitary designs are prevalent across essentially all sub-fields in quantum information. Higher moments are vital for concentration bounds and are intimately related to post-equilibration behavior and complexity growth, but, nevertheless, some applications only leverage the first few moments. To this end, we give improved constants for the design depth and note that the constants given in \cite{brandao_local_2016} are large and could exceed what is required for practical applications. 

As we reviewed, the circuit size $T$ required for local random quantum circuits to form approximate designs, as in \autoref{def:approxdesign}, can be determined from the spectral gap $\Delta(H_{n,t})$ of a frustration-free Hamiltonian. Combining Eq.~\eqref{eq:gHspecgap} with gap amplificiation in Eq.~\eqref{eq:gconvolution}, the depth at which local RQCs form $\ep$-approximate unitary $t$-designs is
\begin{equation}\label{eq:designdepth}
T \geq \frac{n}{\Delta(H_{n,t})} (2nt \log q + \log 1/\varepsilon)\,.
\end{equation}
Furthermore, the gap for local RQCs can be extended to brickwork RQCs using \autoref{lemma:detect}. 

Numerically computing the Hamiltonian gaps for small system sizes, we can then use the Knabe bounds, reviewed in \autorefapp{app:knabe}, to establish design depths for both local and brickwork RQCs, with open and periodic boundary conditions. 
To numerically compute the gaps, we use the Weingarten formalism to construct the local moment operator and numerically diagonalize the resulting Hamiltonian. Details on this procedure are provided in \autorefapp{app:numerics}.
We note that the spectral gaps were investigated numerically using a different method in Ref.~\cite{cwiklinski2013local}, and for all concurrent gaps computed, the results agree.

\subsection*{Explicit low design depths}
We simply give explicit expressions for the design depths for local and brickwork random circuits with open and periodic boundary conditions on $n$ qubits (with $q=2$). The LRQC results are computed from the spectral gaps and the 2-design PRQC results are computed from an exact calculation of the frame potential. 

\begin{table}[h]
\normalsize
\setlength{\tabcolsep}{6pt}
\begin{tabular}{l c c c}
Circuit architecture& 2-designs& 4-designs& 5-designs \\
\hline
local RQCs w/ pbc& $5n(4n+\log1/\varepsilon)$& $3.5n(8n+\log1/\varepsilon)$& $25n(10n+\log1/\varepsilon)$\\
local RQCs w/ obc& $5n(4n+\log1/\varepsilon)$& $4.5n(8n+\log1/\varepsilon)$& $162n(10n+\log1/\varepsilon)$\\
brickwork RQCs w/ pbc& $3.2(2n+\log n+ \log1/\varepsilon)$& $30(8n+\log1/\varepsilon)$& $200n(10n+\log1/\varepsilon)$\\
brickwork RQCs w/ obc& $6.4(2n+\log n+ \log1/\varepsilon)$& $38(8n+\log1/\varepsilon)$& $1288n(10n+\log1/\varepsilon)$
\end{tabular}
\caption{The depths at which local and brickwork random quantum circuits on $n$ qubits ($q=2$) form $\ep$-approximate unitary $t$-designs, for the first few moments, with the best constants taken from analytic and numerical determinations of the spectral gaps (or 2-norm bound).}
\end{table}

For local RQCs, the exact bulk Hamiltonian gap of the second moment for $n=3$ on local qubits is $\Delta(H^{\rm bulk}_{n=3,t=2})=3/5$, as computed in \autoref{theorem:exactt=2} and in agreement with the result in \cite{brandao2010exponential}. For both periodic and open boundary conditions the Knabe bounds (all three Lemmas in \autorefapp{app:knabe}) have the same threshold for subsystem size $n=3$ and give that $\Delta(H_{n,2})\geq 1/5$. The spectral gap of the second moment operator is then $g(\nu_n,2)\leq 1-\frac{1}{5n}$ for both open and periodic local RQCs. 

The 2-design depth for brickwork RQCs above is taken from Ref.~\cite{NHJ19}. In that work, an exact expression is given for the 2-norm of the difference in moment operators $\|M(\nu^{\rm bw}_n,2)-M(\nu_H,2)\|_2$. Converting their result to the strong definition of approximate design in \autoref{def:approxdesign}, and considering both periodic and open boundary conditions yields the above constants. 

We neglect reporting the approximate 3-design depths for different random circuit models because the third moment spectral gaps we computed were in exact agreement with the second moment gaps of the same $q$ and $n$. Moreover, the bounds given by the spectral gaps for the fourth moment are actually stronger and the 4-design depth determined by the Knabe bound is shorter.

For the fourth moment, we must rely on numerical determination of the spectral gaps. Interestingly, the fourth moment Hamiltonian gap is $\Delta(H^{\rm bulk}_{n=3,t=4})=0.5$, up to numerical precision, whereas the Knabe threshold for $n=3$ is $1/2$, and thus we must proceed to larger subsystems. Increasing the subsystem size exceeds the Knabe threshold and gives stronger constants. We employ \autoref{thm:GM} and \autoref{thm:LM} to account for both boundaries conditions, and use the detectability lemma to extend to brickwork RQCs. For the fifth moment, the $n=4$ bulk gap exceeds the threshold for the stronger finite-size criteria in \autoref{thm:GM} and \autoref{thm:LM}, but not for \autoref{thm:knabe}, which gives the above constants. 

One point of interest, the smallest second moment gap $\Delta(H^{\rm bulk}_{n=3,t=2})$ appears to give an asymptotically optimal bound on the gap. For $n=3$, $\Delta(H^{\rm bulk}_{n=3,t=2})=3/5$ and Knabe then gives $\Delta(H_{n,t=2})\geq 1/5$. We can compute the gaps $\Delta(H^{\rm bulk}_{n,t=2})$ for increasing $n$ (in fact, Ref.~\cite{cwiklinski2013local} computed up to $n=21$), which decay as we increase $n$. Fitting the gaps as a function of $n$ suggests that asymptotically $\Delta(H_{n,t=2})\sim 1/5$, which is precisely the lower bound that Knabe gives for the $n=3$ gap.

\vspace*{8pt}
\begin{figure}[h]
    \centering
    \includegraphics[width=0.4\linewidth]{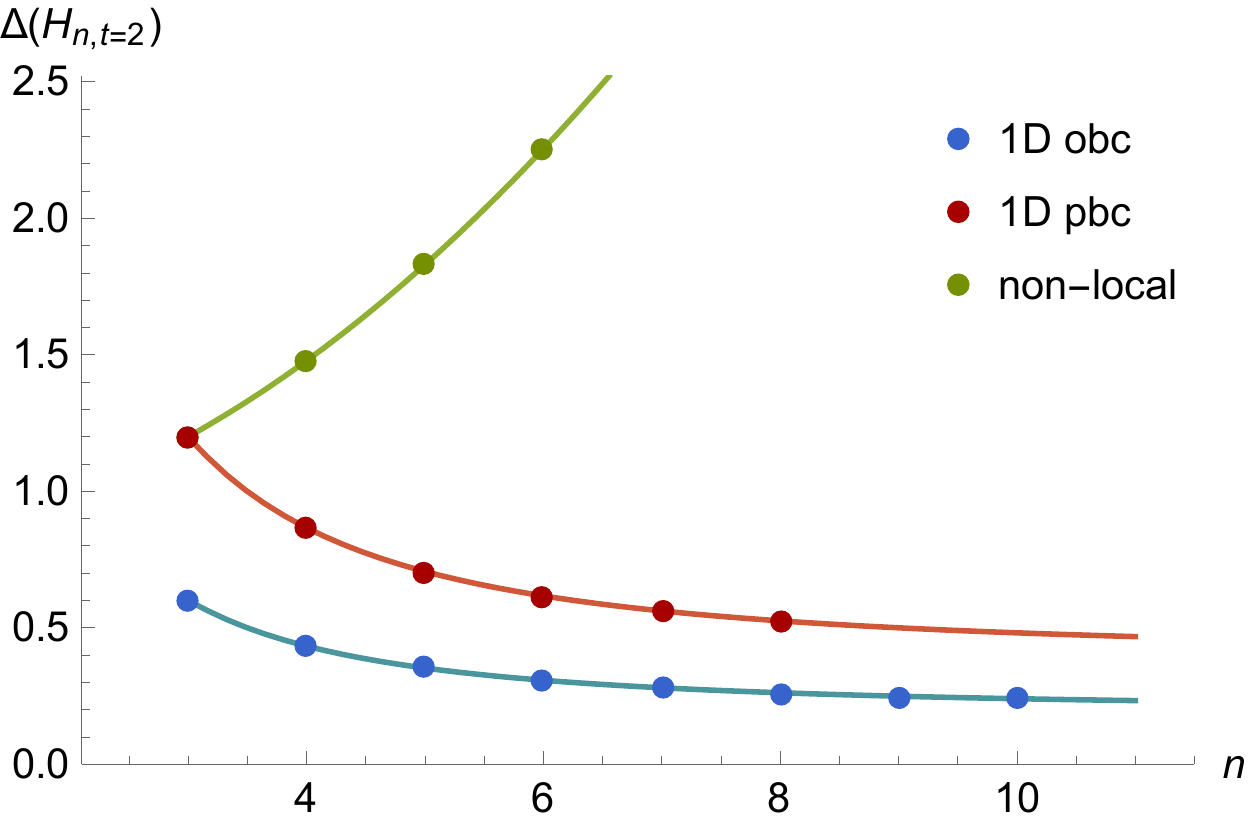}
    \caption{Numerically computed Hamiltonian gaps for the second moment, $t=2$, of obc/pbc $1D$ and non-local RQCs. The best fit scaling suggests a $1/n^2$ decay for the $1D$ spectral gaps, and a (slightly) sub-quadratic growth in the spectral gaps for the non-local Hamiltonian.}
    \label{fig:t2gaps}
\end{figure}
\vspace*{-8pt}

\subsection*{Gap scaling for $1D$ and non-local RQCs}
We conclude with a brief discussion of the gap scaling of the Hamiltonian corresponding to $1D$ RQCs, $H_{n,t} = \sum_i P_{i,i+1}$ with either open or periodic boundary conditions, and that of non-local RQCs, where the Hamiltonian is $H^{\rm non}_{n,t} = \sum_{i<j} P_{i,j}$. Finite-size criteria use gaps of subsystem Hamiltonian with obc to prove lower bounds on the gaps for all $n$; computing the $1D$ pbc and non-local system gaps cannot establish bounds for arbitrary system sizes. Nevertheless, we can still compute the first few nontrivial values to provide evidence for behavior of the spectral gaps in the different random circuit models. 

In \autoref{fig:t2gaps} we show the numerically computed spectral gaps of the second moment, $t=2$, of both the obc and pbc $1D$ Hamiltonians, as well as the non-local case, for increasing system size. A best fit of the obc gaps suggests a $1/n^2$ decay with an asymptotic value of $\sim 0.2$, which is the lower bound given by the $n=3$ Knabe bound. Similarly, a best fit of the $1D$ pbc gaps suggests a $1/n^2$ decay to $\sim 0.4$. For all computed values, the pbc gaps are precisely twice the value of the obc gaps up to numerical error. On the other hand, the non-local Hamiltonian gaps increase as we increase the system size. A polynomial best fit suggests a slightly sub-quadratic scaling of the non-local spectral gaps (the best fit scaling of the first few numerically computed non-local gaps was $n^{1.93}$). If such a scaling exists for higher moments, then this would have implications for the true design depth for non-local RQCs. Whereas the $n$-dependence in the design depth for $1D$ (and higher $D$) RQCs is tight, the $O(n^2 \,{\rm poly}(t))$ design depth we established in \autoref{cor:non-localdesign} could conceivably be improved to $O(n \log(n) \,{\rm poly}(t))$ for non-local RQCs. In fact, a $\Delta(H^{\rm non}_{n,t}) = \Omega(n^2/\log(n))$ scaling of the non-local gaps would be sufficient to prove this, and is consistent with our numerics for the second moment.

\section{Outlook}

A key conjecture made in Ref.~\cite{brandao_local_2016} is whether random quantum circuits in any architecture have a spectral gap that scales as $\Omega(1/\mathrm{poly}(n))$ independent of $t$. 
This would imply approximate unitary designs in depth $O(t \, {\rm poly}(n))$ and, via the results of Ref.~\cite{CompGrowth19}, the Brown-Susskind conjecture~\cite{brown2018second} that local random quantum circuits have a quantum complexity that grows linearly in time for an exponentially long time.

The recursion relation in the proof of \autoref{cor:non-localdesign}, Knabe bounds and also the third condition in the Nachtergaele method~\cite{nachtergaele1996spectral,brandao_local_2016} all rely on an overlap of two projectors acting on different subsets of qudits. 
The fact that for fixed $t$, the overlap becomes small in $q$ suggests to understand the behaviour of $\|(P^{(2)}_{t}\otimes P_{t}^{(1)})(P^{(2)}_{t}\otimes P^{(2)}_{t})-P_{t}^{(3)}\|_{\infty}$ as a function of $q$. 
This can be reformulated as a question about angles between invariant subspaces of subgroups in the unitary group: 
\begin{equation}
\max_t\left\|(P^{(2)}_{t}\otimes P_{t}^{(1)})(P^{(2)}_{t}\otimes P^{(2)}_{t})-P_{t}^{(3)}\right\|_{\infty}\leq \max_{(\pi,V_{\pi})}\cos\left\{V_{\pi}^{SU(q)\otimes SU(q^{2})},V_{\pi}^{SU(q^{2})\otimes SU(q)}\right\},
\end{equation}
where the maximum on the right side is over  all irreducible representations $\pi$ with representation space $V_{\pi}$ except the trivial one.
All eigenvalues of the above operator as well as the eigenvalues of $M(\nu^{\rm bulk}_m)$ are rational functions in $q$ generalizing the solution in \autoref{theorem:exactt=2}. 
This can be seen from the fact that these eigenvalues are solutions to linear systems of equatons with rational functions in $q$ as coefficients.
Unfortunately, characterizing the subspaces $V_{\pi}^{SU(q)\otimes SU(q^{2})}$ in a basis and therefore finding said rational functions seems to be highly non-trivial.
The bound independent of $t$ we obtained in \autoref{lemma:auxgap2} from a convergence result is not sufficient as it converges to $1$ for large $q$.

A possible way to bound the overlap might be via the "dimension trick" in harmonic analysis that was communicated to us by P. Varj\'{u}.
By the Peter-Weyl theorem, all irreps of $SU(D)$ are contained isometrically in the regular representation:
\begin{equation}
L^2(SU(D))\stackrel{\sim}{=}\bigoplus_{\pi} V_{\pi}^{\oplus \dim V_{\pi}}.
\end{equation}
It can be shown that random quantum circuits become absolutely continuous eventually, i.e. there is a density function $\eta\in L^1(SU(q^3))$ such that $\mathrm{d}(\nu^{\rm bw,obc}_n)^{*k_0}=\eta\mathrm{d}\mu_{H}$.
We use the following notation:
\begin{equation}
T_{\rho,\nu}:=\int \rho(U)\mathrm{d}\nu(U)
\end{equation}
for a representation $\rho$ and $T_{\nu}$ for the regular representation.
If one could prove that $\eta\in L^2(SU(q^3))$, then 
\begin{equation}
\tr\left(T_{\nu^{\rm bw,obc}_n}^{2k_0}\right)=\sum_{\pi}\tr\left(T_{\pi,\nu}^{2k_0}\right)\dim(V_{\pi})=\int \eta^2(U)\mathrm{d}\mu_H(U)=:C_{q}<\infty.
\end{equation}
That would imply
\begin{equation}
\|T_{\pi,\nu}\|_{\infty}\leq \left(\frac{C_{q}}{\dim V_{\pi}}\right)^{\frac{1}{2k_0}}.
\end{equation}
For large highest weigths, the dimensions $V_{\pi}$ become arbitrarily small. 
Unfortunately, we do not have any bound on the $L^2$ norm of $\eta$.

\section*{Acknowledgments}
We want to thank Anurag Anshu, Jens Eisert, Sepehr Nezami, Micha{\l} Oszmaniec, and especially P\'{e}ter Varj\'{u} for helpful discussions, as well as Markus Heinrich, Felipe Montealegre-Mora and Ingo Roth for comments on the manuscript.
JH is funded by the Deutsche Forschungsgemeinschaft (DFG, EI 519/14-1).
Research at Perimeter Institute is supported by the Government of Canada through the Department of Innovation, Science and Industry Canada and by the Province of Ontario through the Ministry of Colleges and Universities.

\appendix
\section{Knabe bounds on spectral gaps}\label{app:knabe}
For one-dimensional translational-invariant frustration-free Hamiltonians, we can bound the spectral gap of the system at arbitrary system size using a finite-size criteria, namely that the spectral gap of a small subsystem exceeds a threshold. We refer to these as Knabe bounds. 

In the following, let $n$ denote the global system size and $m$ the local system size. 
We consider $1D$ translationally-invariant Hamiltonians with periodic boundary conditions $H^p_{n}=\sum_{i=1}^n P_{i,i+1}$ and open boundary conditions $H^o_{n}=\sum_{i=1}^{n-1} P_{i,i+1}$. The first criteria relates the gap of an finite-size open Hamiltonian to that of periodic Hamiltonian:
\begin{theorem}[Knabe \cite{knabe1988energy}]\label{thm:knabe}
Let $m\geq 3$ and $n>m$. For a $1D$ frustration-free Hamiltonian with periodic boundary conditions we have that
\begin{equation}
\Delta(H^p_{n}) \geq \frac{m-1}{m-2} \left(\Delta(H^o_{m})-\frac{1}{m-1}\right)\,.
\end{equation}
\end{theorem}
\ni This bound on the spectral gap was later improved to show that:
\begin{theorem}[Gosset-Mozgunov \cite{GM15}]\label{thm:GM}
Let $m\geq 3$ and $n>2m$.
For a $1D$ frustration-free Hamiltonian with periodic boundary conditions we have that
\begin{equation}
\Delta(H^p_{n}) \geq \frac{5}{6}\frac{m^2+m}{m^2-4} \left(\Delta(H^o_{m})-\frac{6}{m(m+1)}\right)\,.
\end{equation}
\end{theorem}
\ni The finite-size criteria were generalized to bound the gap of Hamiltonians with open boundary conditions:
\begin{theorem}[Lemm-Mozgunov \cite{LM18}]\label{thm:LM}
Let $m\geq 3$ and $n\geq 2m$. 
For a $1D$ frustration-free Hamiltonian with open boundary conditions we have that
\begin{equation}
\Delta(H^o_{n}) \geq F(m) \left(\min_{3\leq m'\leq m} \Delta(H^o_{m'})-\frac{G(m)}{m^{3/2}}\right)\,,
\end{equation}
where $F(m)$ and $G(m)$ are known functions of the subsystem size $m$, defined explicitly in \cite{LM18}.
\end{theorem}

In all three bounds, for the case of subsystem size $m=3$, the threshold becomes $1/2$ and the bound on the spectral gap is $\Delta(H_n) \geq 2(\Delta(H^o_3)-1/2)$. In \autoref{thm:LM}, the two functions of $m$ are defined in Ref.~\cite{LM18}, and asymptote to $G(m) \sim 2\sqrt{6}$ and $F(m)\sim \frac{5}{\sqrt{6m}}$.

\section{Orthogonal random circuits and designs for \ensuremath{O(d)}}\label{app:ORQCs}

In this appendix, we briefly generalize some of out results to prove that random quantum circuits constructed out of 2-local gates drawn randomly with respect to the Haar measure on the orthogonal group $\orth(q^2)$, form approximate orthogonal designs. 
The quantum information literature on orthogonal designs and orthogonal RQCs is somewhat sparse. Ref.~\cite{realrandomized2018} studied (exact) orthogonal designs in the context of randomized benchmarking and Ref.~\cite{nhj2018opgrowth} computed the spreading of a local operator under evolution by an orthogonal random quantum circuit.

As the discussion closely follows that of unitary random circuits and unitary designs, our exposition will be succinct. For a probability distribution $\nu$ on the orthogonal group $\orth(d)$, the $t$-fold channels are simply $\Phi^{(t)}_\nu(A) = \int O^{\otimes t}(A)(O^T)^{\otimes t} \d \nu(O)$. Similar to \autoref{def:approxdesign}, we define an approximate orthogonal design as follows:
\begin{definition}[Approximate orthogonal designs]
    A probability distribution $\nu$ on $\orth(d)$ is an $\ep$-approximate orthogonal $t$-design if the $t$-fold channels obey
    \begin{equation}
        \big\| \Phi^{(t)}_\nu - \Phi^{(t)}_{\mu_O} \big\|_\diamond \leq \frac{\ep}{d^t}\,,
    \end{equation}
    where $\mu_O$ denotes the Haar measure on the orthogonal group. Furthermore, we say a probability distribution $\nu$ is a (relative) $\ep$-approximate orthogonal $t$-design if $(1-\ep) \Phi^{(t)}_\nu \preccurlyeq \Phi^{(t)}_{\mu_O} \preccurlyeq (1+\ep)\Phi^{(t)}_\nu$.
\end{definition}
The $t$-th moment operators for a probability distribution $\nu$ on the orthogonal group $\orth(d)$, defined as the vectorization of the $t$-fold channels
\begin{equation}
    M(\nu,t) := {\rm vec}\big(\Phi^{(t)}_\nu\big) = \int O^{\otimes 2t}\, \d \nu(O)\,,
\end{equation}
have a spectral gap given as the operator norm of the difference in moment operators
\begin{equation}
    g_O (\nu,t) := \big\| M(\nu,t) - M(\mu_O,t) \big\|_\infty\,.
\end{equation}
Due to the left/right invariance of the Haar measure, it follows that the orthogonal moment operator is a projector, and thus that $g_O(\nu,t)$ can be amplified as $g_O(\nu^{* k},t) \leq g_O(\nu,t)^k$. The same relation of the spectral gap to the approximate design condition holds in the orthogonal case, namely, for some probability distribution $\nu$ if $g_O(\nu,t)\leq \ep/d^{2t}$, then $\nu$ is an $\ep$-approximate orthogonal $t$-design. 

Consider local random quantum circuits on a $1D$ chain of $n$ qudits with local dimension $q$, where we apply a 2-site orthogonal gate drawn from $\orth(q^2)$ to a nearest-neighbor pair of qudits at each time step. The convergence of orthogonal RQCs to approximate orthogonal designs again follows from a bound on the spectral gap of the moment operators.

\begin{theorem}[Orthogonal spectral gaps for large $q$]\label{thm:Ogaps} Local orthogonal random quantum circuits on $n$ qudits have a spectral gap bounded as
\begin{equation}
    g_O(\nu_n,t)\leq 1- \frac{1}{3n}
\end{equation}
for local dimensions $q\geq 8t^2$, and for all $n\geq 4$ and $t\geq 1$. 
\end{theorem}

As the relation between the spectral gap of the moment operators and the $t$-fold channels is the same as in the unitary case, it is then an immediate corollary that local orthogonal random quantum circuits form approximate orthogonal designs at large $q$. Moreover, the detectablity lemma (\autoref{lemma:detect}) extends the result to brickwork circuits comprised of random orthogonal gates.
\begin{corollary}\label{cor:odesigns}
For any $t\geq 1$ and $n\geq 4$, and for local dimension $q\geq 8t^2$, it holds that
    \begin{enumerate}
    \item  Local orthogonal random quantum circuits of depth $3n(2nt\log(q)+\log(1/\varepsilon))$ are $\varepsilon$-approximate orthogonal $t$-designs.
    \item Brickwork orthogonal random quantum circuits of depth $26(2nt\log(q)+\log(1/\varepsilon))$ are $\varepsilon$-approximate orthogonal $t$-designs.
	\end{enumerate}
\end{corollary}

In order to prove \autoref{thm:Ogaps}, we can express the spectral gap of the orthogonal moments operators as the gap of a frustration-free local Hamiltonian. 
First, we define a convenient short-hand for the orthogonal Haar projector on $m$ qudits $P^{(m)}_O := M(\mu_O,t)$.
Now consider the following Hamiltonian consisting of local nearest-neighbor interaction terms 
\begin{equation}
    H^O_{n,t} = \sum_i P_{i,i+1} \quad\text{with}\quad P_{i,i+1} := \iden -  \iden_{[1,i-1]} \otimes P^{(2)}_O \otimes \iden_{[i+2,n]}\,.
\end{equation}
This $1D$ translationally-invariant Hamiltonian is frustration-free, where the zero-energy ground states are generalizations of those built from permutations as in the unitary case. 

First, let $M_{2t}$ denote the set of all pair partitions on $2t$ elements.
A pair partition $\sigma\in M_{2t}$ is a partition of the set $\{1,\ldots, 2t\}$ into pairs, written as $\{ \{\sigma(1), \sigma(2)\},\ldots,\{\sigma(2t-1), \sigma(2t)\}$, where $\sigma(2n-1)< \sigma(2n)$ and $\sigma(1)<\sigma(3)<\ldots<\sigma(2t-1)$. For example, the set of pair partitions of 4 elements, $M_4$, contains three elements
\begin{equation}
\big\{ \{1,2\}, \{3,4\}\big\}\,, \quad  \big\{ \{1,4\}, \{2,3\}\big\}\,, \quad  \big\{ \{1,3\}, \{2,4\}\big\}\,.
\end{equation}
In general, $M_{2t}$ contains $(2t)!/(2^t t!)$ elements. The set of pair partitions can be simply realized as a subset of the symmetric group $S_{2t}$. Moreover, pair partitions are representatives of the left cosets of the hyperoctahedral group in the symmetric group.

In the $2t$-fold space $(\C^q)^{\otimes 2t}$, let $\ket{\Omega_{nm}}$ be a maximally entangled state on two tensor factors $\ket{\Omega_{nm}} = \frac{1}{\sqrt{q}} \sum_i \ket{i_n i_m}$. Given a pair partition $\sigma\in M_{2k}$, we construct a state as
\begin{equation}
\ket{\varphi_\sigma} := \bigotimes_{j=1}^t \ket{\Omega_{\sigma(2j-1),\sigma(2j)}}\,.
\end{equation}
For any $\sigma\in M_{2t}$, the action of the projector is $P_O \ket{\varphi_\sigma} = \ket{\varphi_\sigma}$. It then follows that the ground states of the Hamiltonian $H_{n,t}$ are the zero energy states $\ket{\varphi_\sigma}^{\otimes n}$ for all $\sigma \in M_{2t}$. 
The $\spann\{ \ket{\varphi_\sigma}\}$ is the zero energy eigenspace of the Hamiltonian, and where $\dim \!\ker H_{n,t} = (2t)!/(2^t t!)$. This can be seen as a consequence of Schur-Weyl duality for the orthogonal group, given by the action of the Brauer algebra which has a basis formed by pair partitions; see \cite{Collins04} and references therein.

Given two pair partitions $\sigma, \tau \in M_{2t}$, we can define an inner product between them as follows. First, define a graph $\mathfrak{g}(\sigma,\tau)$ with vertices $\{1,\ldots,2t\}$ and edges $\{\sigma(2j-1),\sigma(2j)\}_{i=1}^t$ as well as $\{\tau(2j-1),\tau(2j)\}_{i=1}^t$. Let $\ell(\sigma, \tau)$ be the number of connected components of the graph $\mathfrak{g}(\sigma,\tau)$. This allows us to write the Hilbert-Schmidt inner product between two states as $\vev{\varphi_\sigma|\varphi_\tau} = d^{\ell(\sigma,\tau)}$. Note that the diagonal elements with $\sigma=\tau$ will always be $d^t$. 

We define the frame operator for the (non-orthonormal) basis of states $\{\ket{\varphi_\sigma}, \sigma\in M_{2t}\}$
\begin{equation}
    S' = \sum_{\sigma\in M_{2k}} \varphi_\sigma\,,
\end{equation}
where $\varphi_\sigma := \ketbra{\varphi_\sigma}$, and prove the following Lemma:
\begin{lemma}\label{lem:pairpbound} For $d \geq t^2$, the Haar projector $P_O$ on the orthogonal group $\orth(d)$ obeys the following bound
\begin{equation}
   \big\| P_O - S'\big\|_\infty \leq \frac{2t^2}{d}\,.
\end{equation}
\end{lemma}

\begin{proof}
    We start by showing that, similar to the almost orthogonality of permutations in the unitary case, the states $\ket{\varphi_\sigma}$ are nearly orthogonal at large dimension, by upper bounding the sum over inner products of the states $\sum_\sigma |\vev{\varphi_\sigma| \varphi_\tau}|$.
    First, we note that the sum can be expressed in terms of the inner product between pair partitions as
    \begin{equation}
        \sum_{\sigma\in M_{2t}} |\vev{\varphi_\sigma| \varphi_\tau}| = \frac{1}{d^k}\sum_{\sigma\in M_{2t}} d^{\ell(\sigma,\tau)}\,,
        \label{eq:pairpsum}
    \end{equation}
    where $\ell(\sigma,\tau)$ is the number of connected components in the graph defined by $\sigma$ and $\tau$ (and is equivalently the coset-type of the product permutation).
    Ref.~\cite[Eq.~(4.5)]{CollinsMat09} gave an expression for $d^{\ell(\sigma,\tau)}$ in terms of the so-called zonal spherical functions and zonal polynomial (see \cite[Sec.~VII]{MacDonaldHall} for a review). Assuming $d\geq t$, for two pair partitions $\sigma, \tau\in M_{2t}$ we can write 
    \begin{equation}
        d^{\ell(\sigma,\tau)} = \frac{2^t t!}{(2t)!} \sum_{\lambda \vdash t} f_{2\lambda} Z_\lambda(d) \omega_\lambda(\sigma^{-1}\tau) \quad{\rm with}\quad Z_\lambda(d) = \prod_{(i,j)\in \lambda} (d+2j-i-1)
    \end{equation}
    where we sum over integer partitions $\lambda$ of $t$, $f_{2\lambda}$ is the dimension of the irrep associated to $2\lambda$, $\omega_\lambda(\sigma)$ is the zonal spherical function, expressible as a sum of irreducible characters of $S_{2t}$ (see \cite{MacDonaldHall}), 
    and where the zonal polynomial $Z_\lambda(1^d)$ is a symmetric polynomial defined above, with the product taken over the coordinates of the Young diagram of $\lambda$. Using an orthogonality relation between the functions $\omega_\lambda$ (\cite[Eq.~(5.4)]{CollinsMat09}), it follows that $\frac{2^t t!}{(2t)!}\sum_{\sigma\in M_{2t}} \omega_\lambda(\sigma^{-1}\tau) = \delta_{\lambda,\{t\}}$, i.e. sum is non-zero only for the irrep labeled by $\{t\}$. Computing the sum in Eq.~\eqref{eq:pairpsum}, we find
    \begin{equation}
        \sum_{\sigma\in M_{2t}} |\vev{\varphi_\sigma| \varphi_\tau}| = \frac{1}{d^t}\sum_{\sigma\in M_{2t}} d^{\ell(\sigma,\tau)} = \frac{1}{d^t} \prod_{j=1}^{t} (d+2(j-1))\,.
    \end{equation}
    Taking $t^2\leq d$, it then follows that for any fixed pair partition $\tau\in M_{2t}$
    \begin{equation}
        \sum_{\sigma\in M_{2t}} |\vev{\varphi_\sigma| \varphi_\tau}| \leq 1 + \frac{2t^2}{d}\,.
    \end{equation}
    With this bound on the almost-orthogonality of the ground states, the remainder of the proof closely follows \cite[Lem.~17]{brandao_local_2016}.
    Defining the synthesis operator $B' := \sum_\sigma \ket{\sigma}\!\bra{\varphi_\sigma}$ for the orthonormal basis $\{\ket{\sigma}\}$ of ${\rm span}\{\ket{\varphi_\sigma}\}$, where $B'{}^\dagger B' = S'$, and noting that $B'{}^\dagger B'$ and $B' B'{}^\dagger$ have the same eigenvalues, we can then bound the operator norm difference of $S'$ and the Haar projector as
    \begin{equation}
        \|S' - P_O\|_\infty = \Big\| B' B'^\dagger - \sum_\sigma \ketbra{\sigma} \Big\|_\infty = \max_\sigma \sum_{\tau\neq \sigma} |\vev{\varphi_\sigma| \varphi_\tau}|\leq \frac{2t^2}{d}\,.
    \end{equation}
\end{proof}
Using \autoref{lem:pairpbound}, we can proceed with a bound on the orthogonal spectral gap at large local dimension, completely analogous to \autoref{thm:gapbound}.
\begin{proof}[Proof of \autoref{thm:Ogaps}]
Consider the probability distribution $\nu^{\rm bulk}_3$ defined as the application of a single Haar-random orthogonal gate from $\orth(q^2)$ on a random nearest-neighbor pair of 3 qudits, i.e.\ either on qudits 1 and 2 or 2 and 3. We want to bound the operator norm of the difference of moment operators
\begin{equation}
    M\big(\nu_3^{\rm bulk},t\big) - M(\mu_O,t) = \frac{1}{2}\left( P_O^{(2)} \otimes \iden + \iden \otimes P_O^{(2)}\right) - P_O^{(3)}\,.
\end{equation}
Using the operator $B' := \sum_\sigma \ket{\sigma}\!\bra{\varphi_\sigma}$ for the orthonormal basis $\{\ket{\sigma}\}$ of ${\rm span}\{\ket{\varphi_\sigma}\}$, we apply \autoref{lem:pairpbound} to show
\begin{align}
\Big\|M\big(\nu_3^{\rm bulk},t\big) - M(\mu_O,t)\Big\|_\infty &\leq \frac{1}{2} \bigg\|\sum_\sigma \left(\varphi_{\sigma}^{\otimes 2}\otimes \iden+\iden\otimes \varphi_{\sigma}^{\otimes 2}\right)-2\varphi_{\sigma}^{\otimes 3}\bigg\|_\infty + \frac{2t^2}{q^2}+\frac{2t^2}{q^3}\nn
&\leq \frac12 \|B'^\dagger B'\|_\infty \max_\sigma \big\|\varphi_\sigma \otimes \iden + \iden\otimes \varphi_\sigma - 2\varphi_\sigma\otimes\varphi_\sigma \big\|_\infty + \frac{2t^2}{q^2}+\frac{2t^2}{q^3}\nn
&\leq \frac12\left(1+\frac{2t^2}{q}\right)+\frac{2t^2}{q^2}+\frac{2t^2}{q^3}\,.
\end{align}
Taking $q\geq 8t^2$, we then find that for any $t\geq 1$
\begin{equation}
    \Big\|M\big(\nu_3^{\rm bulk},t\big) - M(\mu_O,t)\Big\|_\infty \leq \frac{2}{3}\,.
\end{equation}
We can re-express the bound on the norm of the difference in moment operators for orthogonal RQCs as a bound on the spectral gap of the frustration-free Hamiltonian
\begin{equation}
    \Delta\big(H^{O,{\rm bulk}}_{3,t}\big) \geq \frac{2}{3}\,.
\end{equation}
Using the Knabe bound in \autoref{lemma:knabe} for subsystem size $m=3$, we conclude that
\begin{equation}
    \Delta(H^O_{n,t}) \geq 2\left(\frac{2}{3}-\frac{1}{2}\right) = \frac{1}{3}\,.
\end{equation}
As $g_O(\nu_n,t) = 1-\Delta(H^O_{n,t})/n$, the claim then follows.
\end{proof}

As we discussed, \cite{brandao_local_2016} proved a lower bound on the spectral gap using the path-coupling method, specifically a version for random walks on the unitary group \cite{oliveira2009convergence}. We conclude by noting that path-coupling in the orthogonal case should also give an exponentially small (albeit $t$-independent) lower bound on the spectral gap of $\Delta(H^O_{n,t})$, which, combined with the Nachtergaele method, then would prove that local orthogonal random quantum circuits form approximate orthogonal $t$-designs in $O({\rm poly}(t) n^2)$ depth. 
We leave this investigation to future work.

\subsubsection*{Numerical gaps for orthogonal RQCs}
By explicitly constructing the orthogonal moment operator $P_O$ in the Weingarten formalism, as described in \autorefapp{app:numerics}, we can then numerically determine the spectral gaps for orthogonal RQCs. Applying Knabe bounds for both periodic and open boundary conditions, and using the detectability lemma to extend to brickwork circuits, we find convergence to $\ep$-approximate orthogonal designs in the following circuit depths:
\begin{center}
\setlength{\tabcolsep}{12pt}
\begin{tabular}{l c c}
Circuit architecture& 2-designs& 3-designs \\
\hline
local ORQCs w/ pbc& $7n(4n+\log 1/\ep)$& $8n(6n+\log1/\ep)$\\
local ORQCs w/ obc& $9n(4n+\log 1/\ep)$& $9n(6n+\log1/\ep)$\\
brickwork ORQCs w/ pbc& $55(4n+\log 1/\ep)$& $66(6n+\log1/\ep)$\\
brickwork ORQCs w/ obc& $73(4n+\log 1/\ep)$& $73(6n+\log1/\ep)$
\end{tabular}
\end{center}
As the set of pair partitions $M_{2t}$ grows substantially faster than permutations, numerical determination of the spectral gaps for higher moments quickly became computationally intractable. We were able to compute gaps for the fourth moment, but none that exceeded the Knabe threshold, and thus we just report design depths for the second and third moments. We further note that, unlike in the unitary case, the smallest nontrivial second moment gap $\Delta(H_{n=3,t=2}^{O,{\rm bulk}})$ did not give optimal design depths for all $n$, and lower bounds on $\Delta(H_{n,t=2}^O)$ improved as we increased the subsystem size.

We also computed the spectral gaps for the simplest nontrivial moment operator, with $n=3$ and $t=2$, for varying local dimensions ($q=2$ up to $q=6$). 
In analogy to \autoref{theorem:exactt=2}, we subsequently conjecture that:
\begin{conjecture}[Exact orthogonal gaps for $n=3$ and $t=2/3$]
The spectral gaps for the bulk orthogonal Hamiltonian, where $H^{O, {\rm bulk}}_{n,t} = \sum_{i=1}^{n-1} P^O_{i,i+1}$\,, for $n=3$ and $t=2$ and $3$ are given by
\begin{equation}
    \Delta\big(H^{O, {\rm bulk}}_{n=3,t=2/3}\big) = 1 - \frac{q (q + 2)}{(q + 1)(q^2 + 2)}\,.
\end{equation}
\end{conjecture}
Combined with the Knabe bound, this would imply a similar result as in \autoref{cor:exactgap}, efficient orthogonal 2-designs on $n$ qudits.

\section{Details on numerics}\label{app:numerics}
In this appendix, we summarize some details on how numerics for evaluating the spectral gaps were preformed. To compute the gaps we first construct the local moment operator $P_H^{(2)}$ in the Weingarten formalism \cite{Collins02,Collins04}. 
We can write the moment operator on $n$ qudits as
\begin{equation}
    P_H^{(n)} = \int U^{\otimes t}\otimes \overline{U}^{\otimes t} \,\d\mu_H(U) = \sum_{\pi, \sigma\in S_t} \Wg(\pi^{-1}\sigma,q^n) \ket{\varphi_\pi}\!\bra{\varphi_\sigma}^{\otimes n}\,,
\end{equation}
where again $\ket{\varphi_\pi} := (\iden \otimes r(\pi)) \ket\Omega $, $r(\pi)$ is the standard representation of the permutation and $\ket\Omega$ is the maximally entangled state on $(\C^q)^{\otimes t}\otimes (\C^q)^{\otimes t} $. The unitary Weingarten function $\Wg(\pi,d)$ is a function of permutations $\pi \in S_t$ and admits an expansion in terms of characters of the symmetric group \cite{Collins04} as follows 
\begin{equation}
    \Wg(\pi,d) = \frac{1}{t!} \sum_{\substack{\lambda\vdash t\\ \ell(\lambda)\leq d}} \frac{f_\lambda\, \chi_\lambda(\pi)}{c_\lambda(d)}\,, \quad{\rm where}\quad c_\lambda(d):= \prod_{(i,j)\in \lambda} (d+j-1)
\end{equation}
and where we sum over integer partitions of $t$, restricting to partitions of length $\ell(\lambda)\leq d$, $f_\lambda$ is the dimension of the irreducible representation labeled by $\lambda$, and $\chi_\lambda(\pi)$ is the irreducible character of $\lambda$ on the permutation $\pi\in S_t$. Lastly, $c_\lambda(d)$ is a polynomial (related to the Schur polynomial) where the product above is taken over coordinates of the Young diagram corresponding to $\lambda$. 

Using the above formulation, we can numerically construct the moment operators $P_H^{(2)}$ by computing the unitary Weingarten functions, from which we can then construct the Hamiltonian $H_{n,t}$. Doing so, we can numerically compute the first few eigenvalues of the resulting sparse matrix using power methods. Specifically, the Lanczos algorithm efficiently finds the eigenvalues of interest and allows us to determine the spectral gap of the Hamiltonian. 

Numerics for the orthogonal gaps can also be done using the Weingarten formalism for the orthogonal group \cite{Collins04,CollinsMat09}. The orthogonal moment operator on $n$ qudits can be written as
\begin{equation}
    P^{(n)}_O = \int O^{\otimes 2t}\, \d \mu_H(O) = \sum_{\sigma,\tau\in M_{2t}} \Wg^O(\sigma^{-1}\tau, q^n) \ket{\varphi_\sigma}\!\bra{\varphi_\tau}^{\otimes n}\,,
\end{equation}
where we sum over pair partitions and $\ket{\varphi_\sigma}$ are the states defined in the previous appendix as representations of pair partitions acting on maximally entangled states in the $2t$-fold space. In the equation above, $\Wg^O(\sigma,d)$ is the orthogonal Weingarten function on a pair partition $\sigma\in M_{2t}$. Like in the unitary case, the orthogonal Weingarten function admits an expansion in terms characters \cite{CollinsMat09} as
\begin{equation}
    \Wg^O(\sigma,d) = \frac{2^t t!}{(2t)!} \sum_{\substack{\lambda\vdash t\\ \ell(\lambda)\leq d}} \frac{f_{2\lambda}\, \omega_\lambda(\sigma)}{Z_\lambda(d)}\,,
\end{equation}
where again we sum over integer partitions of $t$, $f_{2\lambda}$ is the dimension of the $2\lambda$ irrep of $S_{2t}$, $\omega_\lambda(\sigma)$ is the zonal spherical function, and $Z_\lambda(d)$ is polynomial in $d$, both defined in \autorefapp{app:ORQCs}. 

Again, we can numerically construct the local orthogonal moment operators $P_O^{(2)}$ by computing the orthogonal Weingarten functions, and then the bulk Hamiltonian. Using sparse matrix methods to find the first few eigenvalues gives the desired numerical values of the spectral gaps.

\bibliographystyle{utphys}
\bibliography{note_local_dimension}

\end{document}